\documentclass[preprint,12pt]{imsart}

\RequirePackage[OT1]{fontenc}
\RequirePackage{amsthm,amsmath}
\RequirePackage[numbers]{natbib}
\RequirePackage[colorlinks,citecolor=blue,urlcolor=blue]{hyperref}

\usepackage{amscd,amsfonts,amssymb,amsmath,latexsym,array,hhline,xcolor,graphicx}
\usepackage{float}
\usepackage{appendix}
\usepackage{booktabs}
\usepackage{caption}
\usepackage{epstopdf}

\newcommand\cF{{\cal F}}

\newcommand\cB{{\cal B}}

\newcommand\cP{{\cal P}}

\newcommand\R{{\bf R}}
\newcommand\Q{{\bf Q}}

\newcommand\N{\mathbb{N}}

\newtheorem{theo}{Theorem}[section]
\newtheorem{prop}[theo]{Proposition}
\newtheorem{lemm}[theo]{Lemma}
\newtheorem{defi}[theo]{Definition}
\newtheorem{ex}[theo]{Example}

\newtheorem{coro}[theo]{Corollary}
\newtheorem{rem}[theo]{Remark}
\newcommand\fdem{$\Box$}
\newcommand\beq{\begin{equation}}
\newcommand\eeq{\end{equation}}
\newcommand\beqa{\begin{equation*}}
\newcommand\eeqa{\end{equation*}}
\newcommand\bea{\begin{eqnarray}}
\newcommand\eea{\end{eqnarray}}
\newcommand\bean{\begin{eqnarray*}}
\newcommand\eean{\end{eqnarray*}}

\DeclareMathOperator{\essinf}{ess\;inf}
\DeclareMathOperator{\esssup}{ess\;sup}

\begin{document}

\begin{frontmatter}

\title{Robust  discrete-time super-hedging strategies under AIP condition and  under price uncertainty }

\author[A1]{ Meriam EL MANSOUR, }
\author[A1]{ Emmanuel LEPINETTE}

\address[A1]{ Ceremade, UMR  CNRS 7534,  Paris Dauphine University, PSL National Research, Place du Mar\'echal De Lattre De Tassigny, \\
75775 Paris cedex 16, France and\\
Gosaef, Faculty of Sciences of Tunis, Tunisia.\\
Email: emmanuel.lepinette@ceremade.dauphine.fr
}

\begin{abstract} We solve the  problem of super-hedging European or Asian options  for discrete-time financial market models where executable prices are uncertain. The risky asset prices  are not described by single-valued processes but measurable selections of random sets that allows to consider a large variety of models including bid-ask models with order books, but also models with a delay in the execution of the orders. We provide a numerical procedure to compute the infimum price under a weak no-arbitrage condition, the so-called AIP condition, under which the prices of the non negative European options are non negative. This condition is weaker than the existence of a risk-neutral martingale measure but it is sufficient to numerically solve  the super-hedging problem. We illustrate our method by a numerical example.
 \end{abstract}


\end{frontmatter}

\section{Introduction}
As observed in practice, the executed value of an asset may depend on the order sent by the trader and, also, on the quantities available in the order book. Among  the possible causes of the   well-known slippage  phenomenon,   delays in the execution of the orders,  liquidity disorders,  market impacts, or transaction costs may  influence the executed value. An approach to overcome this difficulty is to assume that we do not know in advance the traded prices. In that case, as proposed in the paper, the order that the trader sends is a mapping that associates to each possible price available in the market a quantity to sell or buy. This is exactly  what we generally observe in practice, in a presence of an order book for example, since there is no single price. 

On the contrary, it is traditional in mathematical finance to suppose that we  first observe  a (new) single  market price and, then, we  choose almost instantaneously  the number of assets to sell or buy in order to revise the portfolio. This means that the last traded price is kept constant long enough in the order book. Moreover,  it coincides with a bid and ask price so that the buy and sell orders are executed at the same value.

 In  real life, there may be  delayed information, see the recent paper \cite{AO} or \cite{OSZ}, \cite{SZ} among others on stochastic control. The delayed information in the problem of pricing is sometimes modeled through incomplete or restricted information as in \cite{IM}, \cite{KSW}, \cite{DKS}, \cite{Da} or using a two filtrations setting as in \cite{CuKT}.

Another type of uncertainty is due to the choice of the model   supposed to approximate the real financial market \cite{BGK}. Model risk may  lead to price misevaluations that are studied in recent papers, in the growing field of robust finance. Since the seminal work of Knight \cite{Kn}, it is now broadly accepted that uncertainty may be described by a parametrized family of models, instead of considering only one model, if there is a lack of information on the parameters, see \cite{PWZ} , \cite{BRS}, \cite{NN}, \cite{BK}, \cite{BK1}, \cite{FNS}, \cite{TTU}. Other models consider that the market is driven by a family of probability measures in such a way that uncertainty stems from the existence of several possible reference probability measures determining which events are negligible, see \cite{Q},  \cite{H}, \cite{CKT}, \cite{BN}, \cite{BBK}, \cite{BFM}, \cite{OW}, \cite{COW}.

In any case, uncertainty is taken into account in the literature by considering either  several probabilistic structures, e.g. a family of reference probability measures and filtrations for the same price process or a family of price process models on the same stochastic basis. In the recent paper \cite{RM}, the choice is made to fix only one filtered probability space on which a collection of stochastic processes describes the possible dynamics of the stock prices.  We follow this alternative approach. Precisely, we consider a unique stochastic basis but we suppose that, in discrete time, the next stock prices at any time are not modeled by a unique vector-valued random variable as it is usual to do. Instead, we assume that the next stock prices belong to a collection of possible processes. The approach we adopt in our paper is slightly different from \cite{RM} in the sense that the collections of possible prices we consider are connected from time to time in such a way that it is possible to represent them through measurable random sets. 

Moreover, a less common type of uncertainty is introduced in this paper. Recall that it is usual in the literature, even in the recent papers on robust finance, to suppose that the transactions are executed at a price which is known in advance. For example, in the Black and Scholes model, the delta-hedging strategy for the European Call option at time $t$ is a function $\Phi(t,S_t)$ of the single price $S_t$ observed at time $t$. In practice, the strategy is discretized at some dates $(t_i)_{i=0,\cdots,n}$ with $n\to +\infty$ so that the number of stocks to trade at time $t_i$ is $\Delta\Phi_{t_i}=\Phi(t_i,S_{t_i})-\Phi(t_{i-1},S_{t_{i-1}})$. In the case where $\Delta\Phi_{t_i}<0$, the executed price at time $t_i$ should be a bid price in the order book and an ask price otherwise, i.e. there should be at least two possible prices. 

We take into account this ambiguity or uncertainty in our paper by assuming that 
there may be several possible executable prices at the next instant. This means in particular that we do not know in advance the price when we send an order to be executed. Precisely, an executed price $S_t$ at time $t$ is only $\cF_{t+1}$-measurable where $\cF_t$ describes the market information available at time $t$.  This is illustrated in our numerical example where  the stock price is modeled by a pair of bid and ask prices.  \smallskip

This article addresses the super-hedging problem of European or Asian options under uncertainty  and may be easily adapted to American options in discrete time.  Here the uncertainty mainly refers to
the uncertainty in executed prices due to the delay, which is modeled by random sets, and
there is  one single physical probability measure.  Moreover, uncertainty may also refers to the presence of an order book so that several  prices may exist and depend on the traded volumes. The advantage of the approach we consider is its flexibility, including a large variety of possible models, e.g. with transaction costs or limit order books. Contrarily to the classical approach, we do not suppose the existence of a risk-neutral probability measure but we work under the AIP condition of \cite{BCL}, \cite{CL}, i.e. we suppose that the super-hedging prices of the non-negative European claims are non-negative, as it is easily observed in the real financial market. We recall that the AIP condition is weaker than the usual NA condition but it is sufficient to deduce numerically tractable pricing estimations, as illustrated in our numerical example. \smallskip

The paper first focuses on the one-period case, see Section \ref{Onestep}, and the multi-period case is automatically obtained by (measurably) paste all periods together. The one-period hedging problem can be described as:

$$V_{t-1}+\theta_{t-1}\Delta S_t\ge g_t(S_0,\cdots,S_t), \,{\rm a.s.}\quad {\rm for\,all\,} S_t\in \Lambda_t((S_u)_{u\le t-1}).$$

Here $S_t$ is a possible executed price which is $\cF_{t+1}$-measurable, $\theta_{t-1}$ is a trading strategy which is made at time $t-1$ and its outcome is revealed at the same  time $t$ as $S_{t-1}$ due to execution delay and, thus, $V_{t-1}$, which models the portfolio value at time $t-1$, is also $\cF_t$-measurable; $g$ is an Asian option to be hedged while $\Lambda_t((S_u)_{u\le t-1})$ represents the set of all possible prices $S_t$ that can be traded strictly after time $t$. The problem is essentially converted to the one without delay by taking supremum conditioned on $\cF_t$ in the above equation, and the (minimal) super-hedging price is provided in Theorem \ref{theo-OS}  in terms of the concave envelope of some related function restricted on the
conditional closure of $\Lambda_t((S_u)_{u\le t-1})$, see \cite{EL}. Properties of the hedging price, including continuity, convexity, and measurability are analyzed in Section \ref{MR}. These properties are important to deduce backwardly the multi-period case which involves a measurable pasting. 

The benefit of our approach is its easy implementation as illustrated in Section \ref{NI}. Indeed, roughly speaking, our main results  state that we only need to know the range of the future price values in terms of the observed  prices to deduce the strategy $\theta_t$ to be followed. This can be achieved  from a historical data.  The strategy  depends at time $t$ on the price $S_t$, i.e. $\theta_t=\theta_t(S_t)$ where $S_t$ is only revealed at time $t+1$ so that the order a time $t$ is the  $\cF_t$-measurable mapping $z\mapsto \theta_t(z)$ and not $\theta_t(S_t)$. Note that the executed price $S_t$ will depend on the model, e.g. $S_t$ may be one of the several bid and ask prices, and the delayed observation of $S_t$  at time $t+1$ allows to deduce the quantity $\theta_t(S_t)$ to hold in the portfolio.

\section{Formulation of the problem}

Let $(\Omega,(\cF_t)_{t \in \{0,\ldots,T+1\}},\cF_T,P)$ be a filtered complete probability space  where $T$ is the time horizon. We suppose that $\cF_0$ is the trivial $\sigma$-algebra and the $\sigma$-algebra $\cF_t$ represents the information available on the market at time $t$.   The financial market we consider is composed of $d$ risky  assets  and a bond $S^0$. We assume without loss of generality that $S^0=1$. 

In the following, we shall consider random subsets $A$ of $\R^d$, i.e. $A=A(\omega)$ may depend on $\omega\in \Omega$. We then denote by $L^0(A,\cF_t)$ the set of all random variables $X_t$ which are $\cF_t$-measurable and satisfies $X_t(\omega)\in A(\omega)$ a.s.. At last, $\R^d_+$ is the set of all $x=(x_i)_{i=1}^d\in \R^d$ such that $x_i\ge 0$ for all $i=1,\cdots,d$.

  \smallskip

Let us consider, for each $t\le T+1$, $\Lambda_t\subseteq L^0(\R^d_+,\cF_{t+1})$  a  collection of $\cF_{t+1}$-measurable  random variables representing the possible executable prices for the risky assets between time  $t$ and time $t+1$. We suppose that,  at time $t$, the set $\Lambda_t$ may depend on the observed traded prices before time $t$, i.e. to each vector of  prices $(S_u)_{u\le t-1}$, we associate a set  $\Lambda_t=\Lambda_t((S_u)_{u\le t-1})$ representing the possible next prices $S_t$ after time $t$ given that we have observed the executed prices  $(S_u)_{u\le t-1}$ at time $t$. We adopt the financial principle that the executed price $S_t$ is only known  strictly after the order is sent at time $t$ but before time $t+1$. 

\begin{defi} \label{PP} A  price process is an $(\cF_{t+1})_{t=-1,\cdots,T}$-measurable non-negative process $(S_{t})_{t =-1,\cdots,T}$ such that  $S_t\in \Lambda_t((S_u)_{u\le t-1})$ is $\cF_{t+1}$ measurable for all $t=0,\cdots,T$ and $S_{-1}\in \R$ is given.
\end{defi}

\begin{ex}
Recall that   $S_{t}$ represents the  prices $(S_t^1,\cdots,S_t^d)$ of $d\ge 1$ risky assets proposed by the market to the portfolio manager when selling or buying. A typical case could be $\Lambda_t=L^0(I_t,\cF_{t+1})$ with 
$$I_t=\Pi_{j=1}^d [S^{bj}_t,S^{aj}_t],$$ where 
$(S^{bj})_{j=1,\cdots,d}$ and $(S^{aj})_{j=1,\cdots,d}$ are respectively the bid and the ask price processes observed in the market between time $t$ and $t+1$ that may depend on $(S_u)_{u\le t-1}$. They are not necessary the best bid/ask prices as, in practice, the real transaction price may be a convex combination of bid and ask prices. Indeed, a transaction is generally the result of an agreement between sellers and buyers but it also depends on the traded volume. Clearly, the portfolio manager does not benefit in general from the last  traded price observed in the market when  sending an order. On the contrary, he should face an uncertain price $S_t$ that depends on the type of order (and may be not executed) but it also depends on some random events he does not control, e.g. slippage. A simple way to model this phenomenon is to suppose that the executable prices obtained by the manager belong to random intervals.\smallskip
\end{ex}

\begin{ex} Another interesting case is when $\Lambda_t=\{S_t^{\theta}:~\theta\in \Theta\}$ is a parametrized family  of random variables. For instance, consider  fixed processes $(\xi_u)_{u\le T}$ and $(m_u)_{u\le T}$  adapted to $(\cF_{t+1})_{t=0,\cdots,T}$ and independent of $\cF_{t}$. Let $C$ be a compact subset of $\R$ and  suppose that $S_{-1}$ is given.  We define recursively 
$$\Lambda_t((S_u)_{u\le t-1})=\left\{S_{t-1}\exp(\sigma \xi_t +m_t):~\,\sigma \in C   \right\},\,S_{t-1}\in \Lambda_{t-1},\,t\le T.$$ In this model, there is an uncertainty on prices because of the unknown parameter (volatility) $\sigma$. This is a classical problem in robust finance, see for example \cite{NN}.\smallskip
 \end{ex}
 
 A  portfolio strategy is an $(\cF_{t+1})_{t =-1,\cdots,T}$-adapted process $\hat\theta=(\theta^0,\theta)$ where, for all $t =0,\cdots,T$, $\theta_{t}\in \R^d$ (resp. $\theta^0_{t}\in \R$)  describes the quantities of  risky assets (resp. the bond) held in the portfolio between time $t$ and time $t+1$. Since the strategies are not supposed to be adapted to $(\cF_t)_{t =0,\cdots,T}$ but only adapted to $(\cF_{t+1})_{t =0,\cdots,T}$, the manager is not supposed  to control the quantity of assets he wants to sell or buy. This is what happens in practice because the orders are not necessarily executed, for instance in the case of limit stock market orders. Precisely, the portfolio manager may send an $\cF_t$-measurable order at time $t$ that depends on the uncertain price $S_{t}$ which is only $\cF_{t+1}$ measurable. For instance, such an order could be {\sl Buy at most $1000$ units at a price less than or equal to $145$ euros} so that the strategies and the executed prices are linked. In the example, the executed quantity  should be deduced from an order book as the minimum between $1000$ and the number of assets we may obtain for a price less that $145$. Then, the executed price is a weighted average of all prices available for less than $145$ in the order book.\smallskip

   For such a strategy $\hat\theta=(\theta^0,\theta),$ we define the portfolio process with initial endowment $V_0\in L^0(\R,\cF_1)$, as the liquidation value
 $$V^{\hat \theta}=\theta^0+\theta S=\theta^0+\sum_{i=1}^d\theta^i S^i.$$
 Recall that  $S_t$ is observed strictly after the portfolio manager  sends an order for  $\theta_t$ at time $t$. In the super-hedging problem we solve, we  expect  orders which are mapping $x\mapsto \Delta \theta_t(x)=\theta_t(x)-\theta_{t-1}((S_u)_{u\le t-1})$ where  $\Delta \theta_t(x)$ is $\cF_t$-measurable and the executed quantity  $\Delta  \theta_t(S_t)$ is only $\cF_{t+1}$-measurable since $S_t$ is $\cF_{t+1}$-measurable. Here the notation $xy$ is used to designate the Euler scalar product between two vectors $x,y$ of $\R^d$.  \smallskip
 
In the following, we only consider self-financing portfolio processes $V^{\hat \theta}$, i.e. they satisfy by definition: $$\Delta V^{\hat \theta}_t:= V^{\hat \theta}_t-V^{\hat \theta}_{t-1}=\theta_{t-1}\Delta S_t,$$
where $\Delta S_t:=S_t-S_{t-1}$. Indeed, this dynamics holds if and only if we have $-(\theta_{t}^0-\theta_{t-1}^0)S^0_t=(\theta_{t}-\theta_{t-1})S_t$. This means that the cost of the new portfolio allocation $(\theta^0_t,\theta_t)$, i.e. buying or selling the quantities $(|\theta_{t}^i-\theta_{t-1}^i|)_{i=0}^d$,  at the executed price $S_t$ is charged to the  cash account.  Therefore, 
\bea \label{self-finan} V^{\hat \theta}_t=V_0+\sum_{u=1}^t  \theta_{u-1} \Delta S_u.\eea
\noindent It is then natural by (\ref{self-finan}) to   write $V^{\theta}=V^{\hat \theta}$. \smallskip

 The aim of the paper  is to solve the following problem:
Construct the minimal super-hedging strategy of an Asian option whose payoff is $g(S_0,\cdots,S_T)$ for some convex deterministic function $g$ on $(\R^d)^{T+1}$. Because of price uncertainty, this means that we shall construct a self-financing strategy $\theta$ and we shall determine the minimal initial endowment $V_0=V^{\theta}_0$  such that we have $V^{\theta}_T\ge g(S_0,S_1,\cdots,S_T)$  independently of the value of the executable prices $S_t\in \Lambda_t((S_u)_{u\le t-1})$ are for $t\le T$. Note that $V_t$ is $\cF_{t+1}$-measurable hence one more step is necessary to deduce the initial endowment $P_0$ at time $t=0$ we need for initiating a super-hedging portfolio process $V$, i.e.  $P_0\ge V_0$. Indeed, $P_0$ should be $\cF_0$-measurable, i.e.  a constant, or equivalently $P_0\ge \esssup_{\cF_0}(V_0)$. We refer to \cite{CL} for the definitions of conditional essential supremum and infimum. 

%To these end we shall   introduce two new concepts  of  conditional random sets: The conditional interior and the conditional closure. The conditional interior is an open version of the conditional core, as recently introduced by L\'epinette and Molchanov, and may be seen as a measurable version of the topological interior. The conditional closure generalizes the concept of conditional support. These concepts are  useful for applications in mathematical finance and conditional optimization.

\section{The  super-hedging problem}\label{SHP}
\subsection{The  one time step resolution }\label{Onestep}

We first introduce the basic tools and theoretical results we need in this section. A set $\Lambda$ of measurable random variables is said $\cF$-decomposable if for any finite partition $(F_i)_{i=1,\cdots,n}\subseteq \cF$ of $\Omega$, and for every family $(\gamma_i)_{i=1,\cdots,n}$ of $\Lambda$, we have $\sum_{i=1}^n\gamma_i1_{F_i}\in \Lambda$. In the following, we denote by $\Sigma(\Lambda)$ the $\cF$-decomposable envelope of $\Lambda$, i.e. the smallest $\cF$-decomposable family containing $\Lambda$. Notice that  
\bean \label{envelopDecomp} \Sigma(\Lambda)=\left\{  \sum_{i=1}^n\gamma_i1_{F_i}:~n\ge 1,\,(\gamma_i)_{i=1,\cdots,n}\subseteq \Lambda,\, (F_i)_{i=1,\cdots,n}\subseteq \cF\, {\rm s.t. } \sum_{i=1}^nF_i=\Omega\right\}.\quad\eean
The closure $\overline{\Sigma}(\Lambda)$ in probability of $\Sigma(\Lambda)$ is decomposable even if $\Lambda$ is not decomposable. By \cite[Theorem 2.4]{LM}, there exists a $\cF$-measurable closed random set $\sigma(\Lambda)$ such that $\overline{\Sigma}(\Lambda)=L^0(\sigma(\Lambda),\cF)$ is the set of all $\cF$-measurable selectors of $\sigma(\Lambda)$. \smallskip

We now introduce the general one step problem between the dates $t-1$ and $t$ for $t\ge 1$. To do so, we suppose that after  time $t-1$ but strictly before time $t$ the portfolio manager observes the price  $S_{t-1}$, as a consequence of her/his order, see Definition \ref{PP}. More precisely,  the portfolio manager   knows $(S_u)_{u\le t-2}$ at time $t-1$ and  sends an order at time $t-1$ which is executed with a delay so that the executed price $S_{t-1}\in \Lambda_{t-1}((S_u)_{u\le t-2})$ is only observed strictly after $t-1$, i.e. $S_{t-1}$ is $\cF_t$-measurable. \smallskip

In the following, we consider the $\sigma$-algebra $\cF_{t}=\sigma(S_u:~u\le t-1)$ for all $t\ge 1$. Let us consider a  random function $g_t$ defined on $(\R^d)^{t+1}$, $t\ge 1$. We assume that the mapping $(\omega,z)\mapsto g_t(S_0(\omega),\cdots,S_{t-1}(\omega),z)$ is $\cF_{t}\times \cB(\R^d)$-measurable and $z\mapsto g_t(S_0,S_1,\cdots,S_{t-1},z)$ is lower-semicontinuous (l.s.c.) almost surely independently  the price process $(S_u)_{u\le t-1}$ is. The first goal is to characterise the set $\cP_{t-1}$ of all $V_{t-1}\in L^0(\R,\cF_{t})$ that depend on $(S_u)_{u\le t-1}$ such that:
\bea \label{IneqSH} V_{t-1}+\theta_{t-1}\Delta S_t\ge g_t(S_0,\cdots,S_t), \,{\rm a.s.}\quad {\rm for\,all\,} S_t\in \Lambda_t((S_u)_{u\le t-1}),\eea
for some $\theta_{t-1}\in L^0(\R^d,\cF_{t})$ \footnote{Note that the condition $V_{t-1}\in L^0(\R,\cF_{t})$ is not sufficient for the portfolio manager to observe it when $t=1$ as $V_{0}$ is not $\cF_{0}$-measurable.}. As $\theta_{t-1}$ is only $\cF_t$-measurable, we also expect a dependence between  $\theta_{t-1}$ and $(S_u)_{u\le t-1}$ as we shall see later.  Nevertheless, we do not suppose an explicit dependence of $\Lambda_t((S_u)_{u\le t-1})$ with respect to $\theta_{t-1}$, which is an open problem.    We observe  by lower-semicontinuity that (\ref{IneqSH}) holds if and only if 
\bea \label{IneqSH+} V_{t-1}+\theta_{t-1}\Delta S_t\ge g_t(S_0,\cdots,S_t),\,{\rm a.s.}\quad {\rm for\,all\,} S_t\in \overline{\Sigma}(\Lambda_t((S_u)_{u\le t-1})).\eea
\noindent This means that we may suppose w.l.o.g. that $\overline{\Sigma}(\Lambda_t((S_u)_{u\le t-1}))=\Lambda_t((S_u)_{u\le t-1})$. In the following, we denote by 
$I_t((S_u)_{u\le t-1})$ the $\cF_{t+1}$-measurable closed random set such that $\overline{\Sigma}(\Lambda_t((S_u)_{u\le t-1}))=L^0(I_t((S_u)_{u\le t-1}),\cF_{t+1})$, see \cite[Theorem 2.4]{LM}. \smallskip

By \cite[Theorem 3.4]{EL}, we deduce that (\ref{IneqSH}) is equivalent to $V_{t-1}\ge p_{t-1}$ where $p_{t-1}= p_{t-1}((S_u)_{u\le t-1},\theta_{t-1})$ is given by
\bean 
p_{t-1} &=&\theta_{t-1}S_{t-1}+\sup_{z\in {\rm cl\,}(I_t((S_u)_{u\le t-1})|\cF_{t})}\left(g_t(S_1,\cdots,S_{t-1},z)-\theta_{t-1}z \right),\\
\quad &=&\theta_{t-1}S_{t-1}+f_{t-1}^*(-\theta_{t-1}).\eean

In the formula above, ${\rm cl\,}(I_t((S_u)_{u\le t-1})|\cF_{t})$ is the conditional closure of $I_t((S_u)_{u\le t-1})$, i.e. the smallest $\cF_{t}$-measurable closed random set which contains $I_t((S_u)_{u\le t-1})$ almost surely. We refer the readers to \cite[Theorem 3.4]{EL} for the existence and uniqueness of such conditional random set. Moreover,  $f_{t-1}^*(y)=\sup_{z\in \R^d}(yz-f_{t-1}(z))$ is the Fenchel-Legendre conjugate function of $f_{t-1}$ defined as
\bea \label{f} f_{t-1}(z)&:=&-g_t(S_0,\cdots,S_{t-1},z)+\delta_{{\rm cl\,}(I_t((S_u)_{u\le t-1})|\cF_{t})}(z),\eea
where $\delta_{{\rm cl\,}(I_t((S_u)_{u\le t-1})|\cF_{t})}\in \{0,\infty\}$ is infinite on the complimentary of \\ ${\rm cl\,}(I_t((S_u)_{u\le t-1})|\cF_{t})$ and $0$ otherwise. Notice that $f_{t-1}^*$ is convex and l.s.c. as a supremum (on ${\rm cl\,}(I_t((S_u)_{u\le t-1})|\cF_{t})$) of convex and l.s.c. functions. Moreover, by \cite[Theorem 3.4]{EL}, $(\omega,y)\mapsto f_{t-1}^*(\omega,y)$ is $\cF_{t}\times \cB(\R^d)$-measurable.  Therefore, ${\rm Dom\,}f_{t-1}^*:=\{y:~f_{t-1}^*(\omega,y)<\infty\}$ is an $\cF_{t}$-measurable random set. We deduce that the $\cF_{t}$-measurable prices at time $t-1$ are given by the Minkowski sum
\bea \nonumber  \cP_{t-1}((S_u)_{u\le t-1})&=&\left\{\theta_{t-1}S_{t-1}+f_{t-1}^*(-\theta_{t-1}):~ \theta_{t-1}\in L^0(\R^d,\cF_{t})   \right\}\\
\label{prices} &&+L^0(\R_+,\cF_{t}).\eea
The second step is to determine the infimum super-hedging price as 
\bea \label{infPrice} p_{t-1}((S_u)_{u\le t-1})=\essinf_{\cF_{t}}\cP_{t-1}((S_u)_{u\le t-1}).\eea To do so, we use  the arguments of \cite[Theorem 2.8]{CL} and we obtain our first main result:

\begin{theo} \label{theo-OS} Suppose that the mapping $(\omega,z)\mapsto g_t(S_0(\omega),\cdots,S_{t-1}(\omega),z)$ is $\cF_{t}\times \cB(\R^d)$-measurable and $z\mapsto g_t(S_0,S_1,\cdots,S_{t-1},z)$ is lower-semicontinuous (l.s.c.) almost surely whatever the price process $(S_u)_{u\le t-1}$ is. Let us consider the function $f_t$ defined by (\ref{f}) and the set of all prices given by (\ref{prices}). Then, the infimum price given by (\ref{infPrice}), satisfies  $ p_{t-1}((S_u)_{u\le t-1})=-f_{t-1}^{**}(S_{t-1})$.\end{theo}

{\sl Proof.} This is a consequence of the following chain of equalities:
\bea \nonumber p_{t-1}((S_u)_{u\le t-1})&=&\essinf_{\cF_{t}} \left\{ \theta_{t-1}S_{t-1}+f_{t-1}^*(-\theta_{t-1}):~ \theta_{t-1}\in L^0(\R^d,\cF_{t})\right\},\\ \nonumber
&=&\essinf_{\cF_{t}}\left\{ -\theta_{t-1}S_{t-1}+f_{t-1}^*(\theta_{t-1}):~ \theta_{t-1}\in L^0(\R^d,\cF_{t})\right\},\\ \nonumber
&=&-\esssup_{\cF_{t}} \left\{ \theta_{t-1}S_{t-1}-f_{t-1}^*(\theta_{t-1}):~ \theta_{t-1}\in L^0(\R^d,\cF_{t})\right\},\\ \nonumber
&=&-\esssup_{\cF_{t}} \left\{ \theta_{t-1}S_{t-1}-f_{t-1}^*(\theta_{t-1}):~ \theta_{t-1}\in L^0({\rm Dom\,}f_{t-1}^*,\cF_{t})\right\},\\
\nonumber &=&-\sup_{z\in \overline{{\rm Dom\,}}f_{t-1}^*} \left( zS_{t-1}-f_{t-1}^*(z)\right),\\\nonumber
&=&-\sup_{z\in \R^d} \left( zS_{t-1}-f_{t-1}^*(z)\right),\\
&=&-f_{t-1}^{**}(S_{t-1}) \label{MR-OneStep}.
\eea
\fdem

Note that we do not need to suppose no-arbitrage conditions to establish the very general pricing formula above. It is only based on the lower-semicontinuity and measurability assumptions satisfied by the payoff $g$.

\subsection{Main properties satisfied by the one time step infimum super-hedging price}\label{MR}

The results of this section are the main contribution of our paper. They are needed to propagate the one time step pricing procedure of Section \ref{Onestep} to the multi-period case. In the following, we suppose that, for all price process $(S_u)_{u\le t-1}$, there exists $\alpha_{t-1}\in L^0(\R^d,\cF_{t})$ and $\beta_{t-1}\in  L^0(\R,\cF_{t})$ that may depend on   $(S_u)_{u\le t-1}$ such that 

\bea \label{hypothese} g_t(S_0,\cdots,S_{t-1},x)\le \alpha_{t-1}x+\beta_{t-1},\quad \forall x\in {\rm cl\,}(I_t((S_u)_{u\le t-1})|\cF_{t}).\eea
This is the case  for Asian options whose payoffs are for example of the form $k(S_0+S_1+\cdots+S_{t}-K)^+$, $k\ge 0$. By  \cite{CL}[Theorem 2.8], we know that 

\bea \label{PriceInf} && p_{t-1}((S_u)_{u\le t-1})=\\ \nonumber
&& \inf\left\{\alpha S_{t-1}+\beta:~\alpha x+\beta\ge g_t(S_0,\cdots,S_{t-1},x),\, \forall x\in {\rm cl\,}(I_t((S_u)_{u\le t-1})|\cF_{t})\right\}.\eea

We first establish the following result:\footnote{The notation $\overline{{\rm conv\,}}(A)$ designates the closed convex hull of A, i.e. the smallest convex closed set containing $A$.}

\begin{prop}\label{PropPrixInfini} Let $(S_u)_{u\le t-1}$ be a price process.  Suppose that the mapping $(\omega,z)\mapsto g_t(S_0(\omega),\cdots,S_{t-1}(\omega),z)$ is $\cF_{t}\times \cB(\R^d)$-measurable and the function $z\mapsto g_t(S_0,\cdots,S_{t-1},z)$ is l.s.c. almost surely. If $S_{t-1}\notin \overline{{\rm conv\,}}{\rm cl\,}(I_t((S_u)_{u\le t-1})|\cF_{t})$, then $p_{t-1}(((S_u)_{u\le t-1}))=-\infty$. Moreover,  $p_{t-1}((S_u)_{u\le t-1})\ge g_t(S_0,\cdots,S_{t-1},S_{t-1})$ if $S_{t-1}\in \overline{{\rm conv\,}}{\rm cl\,}(I_t((S_u)_{u\le t-1})|\cF_{t})$. At last, if $g_t(S_0,\cdots,S_{t-1},\cdot )$ is bounded from below by $m_{t-1}\in L^0(\R,\cF_{t})$ on $ {\rm cl\,}(I_t((S_u)_{u\le t-1})|\cF_{t})$,  then we have \\$p_{t-1}(((S_u)_{u\le t-1}))\ge m_{t-1}$ if $S_{t-1}\in \overline{{\rm conv\,}}{\rm cl\,}(I_t((S_u)_{u\le t-1})|\cF_{t})$.
\end{prop}
{\sl Proof.} Suppose that $S_{t-1}\notin \overline{{\rm conv\,}}{\rm cl\,}(I_t|\cF_{t})$ where $I_t=I_t((S_u)_{u\le t-1})$. By the Hahn-Banach separation theorem and a measurable selection argument, there exists a non null  $\alpha^*_{t-1}$ in $L^0(\R^d\setminus \{0\},\cF_{t})$ and $c^1_{t-1},c^2_{t-1}\in L^0(\R^d,\cF_{t})$  such that we have the inequality  $\alpha^*_{t-1}y < c^1_{t-1}<c^2_{t-1}< \alpha^*_{t-1} S_{t-1}$ for all $y\in  {\rm cl\,}(I_t|\cF_{t})$. Multiplying the inequality by a sufficiently large positive multiplier, we may suppose that $\alpha^*_{t-1}( S_{t-1}- y)\ge n$ where $n\in \mathbb{N}$ is arbitrarily chosen. Let us introduce $\tilde \alpha_{t-1}=\alpha_{t-1}-\alpha^*_{t-1}$ and $\tilde \beta_{t-1}^n=\beta_{t-1}+\alpha^*_{t-1} S_{t-1}-n$, $n\ge 1$. By construction, $\alpha_{t-1}x+\beta_{t-1}\le \tilde\alpha_{t-1}x+\tilde\beta_{t-1}^n$ for all  $x\in  {\rm cl\,}(I_t|\cF_{t})$, where $\alpha_{t-1},\beta_{t-1}$ are given in (\ref{hypothese}). It follows that $\tilde\alpha_{t-1}x+\tilde\beta_{t-1}^n\ge g_t(S_1,\cdots,S_{t-1},x),$ for every  $x\in  {\rm cl\,}(I_t|\cF_{t})$. By (\ref{PriceInf}), we deduce that $p_{t-1}\le \tilde\alpha_{t-1}S_{t-1}+\tilde\beta_{t-1}^n$, i.e. $p_{t-1}\le \alpha_{t-1}+\beta_{t-1}-n$. As $n\to \infty$, we deduce that $p_{t-1}=-\infty$.

Suppose that $z\mapsto g_t(S_1,\cdots,S_{t-1},z)$ is a.s. convex and, furthermore, $S_{t-1}\in \overline{{\rm conv\,}}{\rm cl\,}(I_t|\cF_{t})$. By (\ref{PriceInf}), 
$$p_{t-1}((S_u)_{u\le t-1})\ge g_t(S_0,\cdots,S_{t-1},S_{t-1}).$$ 

At last, suppose that  $z\mapsto g_t(S_0,S_1,\cdots,S_{t-1},z)$ is bounded from below by $m_{t-1}\in L^0(\R,\cF_{t})$ on $ {\rm cl\,}(I_t|\cF_{t})$ and $S_{t-1}\in \overline{{\rm conv\,}} {\rm cl\,}(I_t|\cF_{t})$. Then, $S_{t-1}=\lim_{n\to \infty} S_n$ where $S_n\in {\rm conv\,} {\rm cl\,}(I_t|\cF_{t})$, i.e. $S_n=\sum_{i=1}^{J_n}\lambda_{i,n}x_{i,n}$ where $\lambda_{i,n}\ge 0$  with   $\sum_{i=1}^{J_n}\lambda_{i,n}=1$ and $x_{i,n}\in  {\rm cl\,}(I_t|\cF_{t})$ for all $i,n$. Consider $(\alpha,\beta)$ such that $\alpha x+\beta\ge g_t(S_0,\cdots,S_{t-1},x)$ for all $x\in {\rm cl\,}(I_t|\cF_{t})$. Then, $\alpha S_{t-1}+\beta=\lim_{n\to \infty}(\alpha S_n+\beta)$ with 
\bean \alpha S_n+\beta&=&\sum_{i=1}^{J_n} \lambda_{i,n}(\alpha x_{i,n}+\beta)\ge \sum_{i=1}^{J_n} \lambda_{i,n}  g_t(S_1,\cdots,S_{t-1},x_{i,n})\\
&\ge & m_{t-1}.\eean
We deduce that $\alpha S_{t-1}+\beta\ge m_{t-1}$ hence $p_{t-1}\ge m_{t-1}$ by (\ref{PriceInf}). \fdem

\begin{coro}\label{coro-AIP-first}Let $(S_u)_{u\le t-1}$ be a price process.  Suppose that the mapping $(\omega,z)\mapsto g_t(S_0(\omega),\cdots,S_{t-1}(\omega),z)$ is $\cF_{t}\times \cB(\R^d)$-measurable and the function  $z\mapsto g_t(S_0,\cdots,S_{t-1},z)$ is l.s.c. a.s. and  convex or bounded from below by $m_{t-1}\in L^0(\R,\cF_{t})$ on $ {\rm cl\,}(I_t((S_u)_{u\le t-1})|\cF_{t})$. Then, $p_{t-1}((S_u)_{u\le t-1})\ne -\infty$ if and only if $S_{t-1}\in \overline{{\rm conv\,}} {\rm cl\,}(I_t((S_u)_{u\le t-1})|\cF_{t})$. In particular, the infimum super-hedging price of any non negative payoff function is finite if and only if it is non negative or equivalently if $S_{t-1}\in \overline{{\rm conv\,}} {\rm cl\,}(I_t(S_u)_{u\le t-1}|\cF_{t})$.
\end{coro}

As studied in \cite{CL}, the non negativity of the prices for the zero claim or more generally for non negative European call options corresponds to a weak no arbitrage condition (AIP) which is naturally observed in practice. Adapted to our setting, we introduce the following definition:

\begin{defi}
We say that condition AIP holds between $t-1$ and $t$ if the prices at time $t-1$ of the time $t$ zero claim is non negative for every  price process $(S_{u})_{u\le t-1}$. Moreover, we say that the condition AIP holds when AIP holds at any time step.
\end{defi}

As observed in \cite{CL} and above, when AIP fails, the infimum of the zero claim, and more generally of non negative payoffs, may be $-\infty$. In that case, the numerical procedure we develop in this paper is still valid but unrealistic and non-implementable in practice.  By Corollary \ref{coro-AIP-first}, we have:

\begin{coro} \label{coroAIP} The condition AIP holds between $t-1$ and $t$ if and only if $S_{t-1}\in \overline{{\rm conv\,}} {\rm cl\,}(I_t((S_u)_{u\le t-1})|\cF_{t})$ for any price process $(S_u)_{u\le t-1}$, i.e. $$I_{t-1}((S_u)_{u\le t-2})\subseteq \overline{{\rm conv\,}} {\rm cl\,}(I_t(S_u)_{u\le t-1}|\cF_{t}), \,t\ge 1.$$
\end{coro}

In the following, if $g$ is a function defined on $\R^d$ and $D$ is a subset of $\R^d$, we denote by ${\rm conc}(g,D)$ the (relative) concave envelope of $g$ on $D$, i.e. the smallest concave function defined on $\R^d$ which dominates $g$ only on $D$. Observe that $g\le h$ on $D$ is equivalent to $g-\delta_D\le h$ on $\R^d$. Therefore, ${\rm conc}(g,D)$ always exists as soon as $g$ is dominated by an affine function on $D$. 

The following result allows us to compute the infimum price rather easily.
\begin{lemm}\label{conc}Let $(S_u)_{u\le t-1}$ be a price process.  Suppose that the mapping $(\omega,z)\mapsto g_t(S_0(\omega),\cdots,S_{t-1}(\omega),z)$ is $\cF_{t}\times \cB(\R^d)$-measurable and the function $z\mapsto g_t(S_0,\cdots,S_{t-1},z)$ is l.s.c. almost surely. Consider the concave envelope $$h_{t-1}(x)= {\rm conc}\left(g_t(S_0,\cdots,S_{t-1},\cdot),{\rm cl\,}(I_t((S_u)_{u\le t-1})|\cF_{t})\right)(x).$$ Then, \bea \label{aux-pt-1} && p_{t-1}((S_u)_{u\le t-1})\\ \nonumber
&&=\inf\left\{\alpha S_{t-1}+\beta:~\alpha x+\beta\ge h_{t-1}(x),\, {\rm for\,all\,} x\in {\rm cl\,}(I_t((S_u)_{u\le t-1})|\cF_{t})\right\}.
\eea
\end{lemm}
\begin{proof}
By definition, $h_{t-1}$ is the smallest concave function which dominates $g$. We deduce that the set of all affine functions dominating $g$ coincides with the set of all affine functions dominating $h_{t-1}$. By (\ref{PriceInf}) we deduce that (\ref{aux-pt-1}) holds. 
\end{proof}

The following result provides a criterion under which the infimum price is a price:

\begin{prop} \label{Bound+infprice=price}Suppose that $AIP$ holds. Let $(S_u)_{u\le t-1}$ be a price process. Suppose that the mapping $(\omega,z)\mapsto g_t(S_0(\omega),\cdots,S_{t-1}(\omega),z)$ is $\cF_{t}\times \cB(\R^d)$-measurable and $z\mapsto g_t(S_0,\cdots,S_{t-1},z)$ is l.s.c. almost surely. Moreover, suppose that there exists $\alpha_{t-1}\in L^0(\R^d,\cF_{t})$ and $\beta_{t-1}\in L^0(\R,\cF_{t})$ such that $g_t(S_0,\cdots,S_{t-1},z)\le \alpha_{t-1}z+\beta_{t-1}$ for all $z\in \overline{{\rm conv\,}}{\rm cl\,}(I_t((S_u)_{u\le t-1})|\cF_{t})$ and consider the concave envelope
\bea \label{h} h_{t-1}(x)= {\rm conc}\left(g_t(S_0,\cdots,S_{t-1},\cdot),{\rm cl\,}(I_t((S_u)_{u\le t-1})|\cF_{t})\right)(x).
\eea

We have $p_{t-1}((S_u)_{u\le t-1})\in[g_t(S_0,\cdots,S_{t-1},S_{t-1}), \alpha_{t-1}S_{t-1}+\beta_{t-1}]$. Moreover, if the super-differential $\partial h_{t-1}(S_{t-1})\neq \emptyset$ , then  $p_{t-1}((S_u)_{u\le t-1})=h_{t-1}(S_{t-1})$ is a price, i.e. $p_{t-1}((S_u)_{u\le t-1})\in \cP_{t-1}((S_u)_{u\le t-1})$ with the super-replicating strategies $\theta_{t-1}\in \partial h_{t-1}(S_{t-1})$.

\end{prop}
{\sl Proof.}
It is clear by Lemma \ref{conc} that $p_{t-1}((S_u)_{u\le t-1})\ge h(S_{t-1})$ when $S_{t-1}$ belongs to ${\rm cl\,}(I_t((S_u)_{u\le t-1})|\cF_{t})$.  By definition, for all $r_{t-1}\in \partial h(S_{t-1}) \neq \emptyset$, for all $x\in \overline{{\rm conv\,}}{\rm cl\,}(I_t(S_u)_{u\le t-1}|\cF_{t})$,
\bea \label{i1} h(x)\le h(S_{t-1})+r_{t-1}(x-S_{t-1})=:\delta(r_{t-1},x).\eea  Therefore, $p_{t-1}((S_u)_{u\le t-1})\le \delta(r_{t-1},S_{t-1})=h(S_{t-1})$, and finally 
$$p_{t-1}((S_u)_{u\le t-1})=h(S_{t-1}).$$

At last, applying (\ref{i1}) with $x=S_t\in I_t((S_u)_{u\le t-1}) \subseteq {\rm cl\,}(I_t((S_u)_{u\le t-1})|\cF_{t})$, we deduce  that 
$$p_{t-1}((S_u)_{u\le t-1})+r_{t-1}\Delta S_t\ge h(S_t)\ge g_t(S_0,\cdots,S_{t-1},S_t).$$
 Since $x\mapsto g_t(S_0,\cdots,S_{t-1},x)$ is l.s.c., we consider the   following random  set:
\bean G_{t}&:=&\{ (\omega,r_{t-1}):~\delta(r_{t-1},x)\ge g_t(S_0,\cdots,S_{t-1},x),\,\forall x\in \overline{{\rm conv\,}}{\rm cl\,}(I_t(S_u)_{u\le t-1}|\cF_{t})\},\\
&=&\{ (\omega,r_{t-1}):~\delta(r_{t-1},\gamma_{t}^n)\ge g_t(S_0,\cdots,S_{t-1},\gamma_{t}^n),\,\forall n \in \N \},
\eean
where $(\gamma_{t}^n)_{n \geq 1}$ is a Castaing representation of  $\overline{{\rm conv\,}}{\rm cl\,}(I_t(S_u)_{u\le t-1}|\cF_{t})$. Since $G_{t}$ is $\cF_{t}\times \cB(\R^d)$-measurable and $G_{t} \neq \emptyset$ a.s, it admits a measurable selection which is a measurable strategy $\theta_{t}$ for the price $p_{t-1}((S_u)_{u\le t-1})$. \fdem

\begin{rem}
As the function $h_{t-1}$ in (\ref{h}) is concave and finite a.s. on the conditional closure $\overline{{\rm conv\,}}{\rm cl\,}(I_t(S_u)_{u\le t-1}|\cF_{t})$, see proof of Proposition \ref{PropPrixInfini}, the super-differential $\partial h(S_{t-1})$ of $h_{t-1}$ at the point $S_{t-1}$ is not empty when $S_{t-1}$ belongs to the interior of $\overline{{\rm conv\,}}{\rm cl\,}(I_t(S_u)_{u\le t-1}|\cF_{t})$.
\end{rem}

The following result proves the measurability of the infimum super-hedging price $p_{t-1}((S_u)_{u \leq t-1}) $ with respect to $(S_u)_{u \leq t-1}$. To do so, we suppose the existence of a Castaing representation, see \cite{rw}, \cite{LM}.

\begin{prop} \label{MeasPrice} Suppose that  ${\rm cl\,}(I_t((S_u)_{u \leq t-1})|\cF_{t})$ admits a  Castaing representation $(\xi^m_{t})_{m\ge 1}$ where $\xi^m_{t}=x^m((S_u)_{u \leq t-1})$, for all $m\ge 1$, and $x^m$ are Borel functions on $(\R^d)^t$ independent of $(S_u)_{u \leq t-1}$. Then,  there exist a  Borel function $ \phi_{t-1}$ on  $(\R^d)^t$ such that  $p_{t-1}((S_u)_{u\le t-1})=\phi_{t-1}((S_u)_{u\le t-1}).$
\end{prop}

\begin{proof}
Let $(S_u)_{u\le t-1}$ be a price process. We denote by
\[ \mathcal{S}^{(t-1)}=(S_u)_{u\le t-1} \text{ and } \mathcal{I}_{t-1}= {\rm cl\,}(I_t(\mathcal{S}^{(t-1)})|\cF_{t}). \]
Recall that $$ p_{t-1}(\mathcal{S}^{(t-1)})= \underset{(\alpha,\beta)}\inf\left\{\alpha S_{t-1}+\beta:~\alpha x+\beta\ge g_t(\mathcal{S}^{(t-1)},x),\, {\rm for\,all\,} x\in \mathcal{I}_{t-1} \right\}. $$
By assumption $x^m$ is a Borel function on $(\R^d)^{t}$ independent of the price process $(S_u)_{u \leq t-1}.$
So:
\bean
p_{t-1}(\mathcal{S}^{(t-1)})&=& \underset{(\alpha,\beta)} \inf\left\{\alpha S_{t-1}+\beta:~\alpha x^m(\mathcal{S}^{(t-1)})+\beta\ge g_t(\mathcal{S}^{(t-1)},x^m(\mathcal{S}^{(t-1)})),\forall m  \right\}\\
&=& \underset{\alpha}\inf\left\{\alpha S_{t-1}+f_{t-1}^*(-\alpha,\mathcal{S}^{(t-1)}) \right\}
\eean
such that $ f_{t-1}^*(-\alpha,\mathcal{S}^{(t-1)}) =\underset{m}\sup~ \left[ g_t(\mathcal{S}^{(t-1)},x^m(\mathcal{S}^{(t-1)}))-\alpha x^m(\mathcal{S}^{(t-1)}) \right].$

Let us denote $\Q^d= \{ \alpha^n=(\alpha_1^n,...,\alpha_d^n), n \geq 1, \alpha_i^n \in \Q\}$ and  define  the real-valued mapping $\phi_{t-1}$  as \label{eqborel} $\phi_{t-1}(\mathcal{S}^{(t-1)})= \underset{n }\inf\left\{\alpha^n S_{t-1}+f_{t-1}^*(-\alpha^n,\mathcal{S}^{(t-1)}) \right\}.$
 We claim that
\bea \label{eq*}
p_{t-1}(\mathcal{S}^{(t-1)})&=& \phi_{t-1}(\mathcal{S}^{(t-1)}).
\eea
It is clear that $p_{t-1}(\mathcal{S}^{(t-1)}) \leq  \phi_{t-1}(\mathcal{S}^{(t-1)}).$ Conversely,  let $\alpha \in \R^d,$ 	and $\alpha^n \in \Q^d$ a sequence such that for arbitrary fixed $\epsilon \in {\rm int}(\R_{+}^d)$, we have $\alpha^n \geq \alpha$ and $\alpha > \alpha^n - \epsilon$ componentwise. Then, by definition of $f_{t-1}^*$, we have:
\bean
f_{t-1}^*(-\alpha,\mathcal{S}^{(t-1)} )& \geq & g_t(\mathcal{S}^{(t-1)},x^m(\mathcal{S}^{(t-1)}))-\alpha x^m(\mathcal{S}^{(t-1)}), ~ \forall m \geq 1\\
&\ge & g_t(\mathcal{S}^{(t-1)},x^m(\mathcal{S}^{(t-1)}))-\alpha ^n x^m(\mathcal{S}^{(t-1)})\\
&&+ (\alpha ^n-\alpha )x^m(\mathcal{S}^{(t-1)}), ~ \forall m \geq 1.
\eean
Notice that $x^m(\mathcal{S}^{(t-1)}) \in \R_{+}^{d}$ because $x^m(\mathcal{S}^{(t-1)}) \in \mathcal{I}_{t-1}.$
So, 
\bean
f_{t-1}^*(-\alpha,\mathcal{S}^{(t-1)} )& \geq & g_t(\mathcal{S}^{(t-1)},x^m(\mathcal{S}^{(t-1)}))-\alpha ^n x^m(\mathcal{S}^{(t-1)}), ~ \forall m \geq 1, ~ \forall n \geq 1 \\
& \ge  & f_{t-1}^*(-\alpha^n,\mathcal{S}^{(t-1)} ),  ~ \forall n \geq 1.
\eean

Hence,
\bean
\alpha S_{t-1}+f_{t-1}^*(-\alpha )& \geq & \alpha S_{t-1}+f_{t-1}^*(-\alpha^n ), ~ \forall n \geq 1\\
& \geq & \alpha^n S_{t-1}+f_{t-1}^*(-\alpha^n )- \epsilon S_{t-1},~ \forall n \geq 1 \\
& \geq & \alpha^n S_{t-1}+f_{t-1}^*(-\alpha^n )- \epsilon S_{t-1}, ~ \forall n \geq 1\\
& \geq & \phi_{t-1}(\mathcal{S}^{(t-1)}) - \epsilon S_{t-1}.
\eean
As $\epsilon \to 0$, we get $\alpha S_{t-1}+f_{t-1}^*(-\alpha ) \geq  \phi_{t-1}(\mathcal{S}^{(t-1)})$. Therefore, we deduce that
$p_{t-1}(\mathcal{S}^{(t-1)}) \geq  \phi_{t-1}(\mathcal{S}^{(t-1)}).$ Hence,  the equality (\ref{eq*}) holds, which proves that the infimum super-hedging price  $p_{t-1}((S_u)_{u\le t-1})$ is  measurable with respect to the argument $(S_u)_{u\le t-1}$.
\end{proof}

The rest of this section aims to prove that, under some technical conditions,  the mapping $(S_u)_{u \leq t-1} \longmapsto p_{t-1}((S_u)_{u\le t-1})$ is lower-semicontinuous, which is needed to propagate backwardly the numerical procedure of Theorem \ref{infPrice} in the multi-step model.\smallskip

\begin{defi}
We say that the mapping $$ I_t: (S_u)_{u \leq t-1} \longmapsto{\rm cl\,}(I_t((S_u)_{u\le t-1})|\cF_{t})$$
is lower-semicontinous if the following property holds: For all sequence of price processes $((S^n_u)_{u \leq t-1})_{n\ge 1}$ converging a.s. to a process   $(S_u)_{u \leq t-1}$, and for all $z\in {\rm cl\,}(I_t((S_u)_{u\le t-1})|\cF_{t})$, there exists a sequence $(z^n)_{n\ge 1}$  such that $\lim_nz^n=z$ and $z^n \in  {\rm cl\,}(I_t((S^n_u)_{u\le t-1})|\cF_{t})$ for all $n\ge 1$.

\end{defi}

\begin{ex}Suppose that $d=1$ and 
$${\rm cl\,}(I_t((S_u)_{u\le t-1})|\cF_{t})=[m_{t-1}S_{t-1},M_{t-1}S_{t-1}]$$where $m_{t-1}, M_{t-1}\in L^0(\R_+,\cF_{t})$ and $m_{t-1}\le M_{t-1}$. \smallskip

Consider $z\in{\rm cl\,}(I_t((S_u)_{u\le t-1})|\cF_{t})$, i.e. $z=\alpha_tm_{t-1}S_{t-1}+(1-\alpha_t)M_{t-1}S_{t-1}$ where $\alpha_t\in L^0([0,1],\cF_{t})$. Let us define $z^n=\alpha_tm_{t-1}S^n_{t-1}+(1-\alpha_t)M_{t-1}S^n_{t-1}$ for all $n\ge 1$. Then, $z^n\in {\rm cl\,}(I_t((S^n_u)_{u\le t-1})|\cF_{t})$ and 
$$|z^n-z|\le 2M_{t-1}|S^n_{t-1}-S_{t-1}|$$
hence $\lim_nz^n=z$.

\end{ex}

In the following, we define the closed convex random sets

$$E_{t-1}^{\epsilon}((S_u)_{u\le t-1},z)=\bar{B}(0,\epsilon)\cap \left({\rm cl\,}(I_t((S_u)_{u\le t-1})|\cF_{t})-z \right),$$
where  $\bar{B}(0,\epsilon)$ is the closed ball of center $z=0$ and radius $\epsilon>0$. We say that the mapping $z\mapsto E_{t-1}^{\epsilon}((S_u)_{u\le t-1},z)$ is convex if, for all $\alpha\in [0,1]$, and $z_1,z_2\in \R^d$, we have 
$$E_{t-1}^{\epsilon}((S_u)_{u\le t-1},\alpha z_1+(1-\alpha)z_2)\subseteq \alpha E_{t-1}^{\epsilon}((S_u)_{u\le t-1},z_1)+(1-\alpha)E_{t-1}^{\epsilon}((S_u)_{u\le t-1},z_2).$$

Note that this convexity property above is automatically satisfied if $d=1$.

\begin{prop} Consider a payoff function $g_t$ defined on $(\R^d)^{t+1}$ such that, there exists $\alpha_{t-1}\in L^0((\R^d)^{t+1},\cF_t)$ such that  $g_t(x)-g_t(y)\ge \alpha_{t-1}(x-y)$, $x,y\in (\R^d)^{t+1}$. Suppose that  $ I_t: (S_u)_{u \leq t-1} \longmapsto {\rm cl\,}(I_t((S_u)_{u\le t-1})|\cF_{t})$ is lower-semicontinous and that $z\mapsto E_{t-1}^{\epsilon}((S_u)_{u\le t-1},z) $ is convex  for all $(S_u)_{u\leq t-1}$. Then,  $(S_u)_{u \leq t-1} \longmapsto p_{t-1}((S_u)_{u\le t-1})$ is lower-semicontinuous, i.e. $p_{t-1}((S_u)_{u\le t-1})\le \liminf_n p_{t-1}((S^n_u)_{u\le t-1})$ if $((S^n_u)_{u \leq t-1})_{n\ge 1}$ converges a.s. to   $(S_u)_{u \leq t-1}$.
\end{prop}
\begin{proof} Suppose that $((S^n_u)_{u \leq t-1})_{n\ge 1}$ converges a.s. to   $(S_u)_{u \leq t-1}$. By assumption, we know that for all $z\in {\rm cl\,}(I_t((S_u)_{u\le t-1})|\cF_{t})$, there exists a sequence $z^n\in  {\rm cl\,}(I_t((S^n_u)_{u\le t-1})|\cF_{t})$ such that $\lim_n z_n=z$. We may suppose that $|z-z_n|\le \epsilon$ where $\epsilon>0$ is arbitrarily fixed. By assumption, for all $\tilde z \in {\rm cl\,}(I_t((S^n_u)_{u\le t-1})|\cF_{t})$ in the ball $\bar{B}(z,\epsilon)$ of center $z$ and radius $\epsilon$, we have:
\bea \nonumber
g_t((S_u)_{u\le t-1},z)&\le& g_t((S^n_u)_{u\le t-1},\tilde z) +|\alpha_{t-1}|\times|((S_u)_{u\le t-1},z)-((S^n_u)_{u\le t-1},\tilde z)|,\\ \nonumber
g_t((S_u)_{u\le t-1},z)&\le& g_t((S^n_u)_{u\le t-1},\tilde z) +|\alpha_{t-1}|\sup_{u\le t-1}|S_u^n-S_u|+|\alpha_{t-1}|\epsilon,\\
g_t((S_u)_{u\le t-1},z)&\le&h^{(n)}(\tilde z) +|\alpha_{t-1}|\sup_{u\le t-1}|S_u^n-S_u|+|\alpha_{t-1}|\epsilon, \label{ineqSC}
\eea
where $h^{(n)}$ is an arbitrary affine function satisfying $h^{(n)}\ge g_t((S^n_u)_{u\le t-1},\cdot)$ on ${\rm cl\,}(I_t((S^n_u)_{u\le t-1})|\cF_{t})$. Let us define 
$$\bar{h}^{(n)}(z)=\inf_{\tilde z\in \bar{B}(z,\epsilon)\cap  {\rm cl\,}(I_t((S^n_u)_{u\le t-1})|\cF_{t-1})}h^{(n)}(\tilde z) +|\alpha_{t-1}|\sup_{u\le t-1}|S_u^n-S_u|+|\alpha_{t-1}|\epsilon.$$
By convention, we set $\inf \emptyset=-\infty$. Let us show that $\bar{h}^{(n)}$ is concave. To see it, observe that $\tilde z\in \bar{B}(z,\epsilon)\cap  {\rm cl\,}(I_t((S^n_u)_{u\le t-1})|\cF_{t})$ if and only if $\tilde z= z+u$ where $u\in E^n(z)=\bar{B}(0,\epsilon)\cap \left({\rm cl\,}(I_t((S^n_u)_{u\le t-1})|\cF_{t})-z \right).$ Therefore, 
$$\bar{h}^{(n)}(z)=\inf_{u\in E^n(z)}h^{(n)}( z+u) +|\alpha_{t-1}|\sup_{u\le t-1}|S_u^n-S_u|+|\alpha_{t-1}|\epsilon.$$

 Let $z=\lambda z_1+(1-\lambda)z_2$. We only need to consider the case where $E^n(z_1)\ne \emptyset$ and $E^n(z_2)\ne \emptyset$. We deduce that $E^n(z)\ne \emptyset$. Moreover, by assumption, any $u\in E^n(z)$ may be written as $u=\alpha u_1+(1-\alpha) u_2$ where $u_i\in E^n(z_i)$, $i=1,2$. Therefore, 
\bean h^{(n)}(z+u)&=&\alpha  h^{(n)}(z_1+u_1)+(1-\alpha)  h^{(n)}(z_2+u_2),\\
&\ge&\alpha \bar{h}^{(n)}(z_1)+(1-\alpha)\bar{h}^{(n)}(z_2).
\eean
Taking the infimum in the left hand side of the inequality above, we deduce that $\bar{h}^{(n)}(\lambda z_1+(1-\lambda)z_2)\ge \alpha \bar{h}^{(n)}(z_1)+(1-\alpha)\bar{h}^{(n)}(z_2)$, i.e. $\bar{h}^{(n)}$ is concave.

By (\ref{ineqSC}), we deduce that
$p_{t-1}((S_u)_{u\le t-1})\le \bar{h}^{(n)}(S_t)$ for all $h^{(n)}$. As $S^n_{t-1}\in E^n(S_{t-1})$, for $n$ large enough, under AIP, we deduce that 
$$p_{t-1}((S_u)_{u\le t-1})\le h^{(n)}(S^n_{t-1}) +|\alpha_{t-1}|\sup_{u\le t-1}|S_u^n-S_u|+|\alpha_{t-1}|\epsilon.$$
Taking the infimum over all affine functions $h^{(n)}$, we get that for $n$ large enough:
$$p_{t-1}((S_u)_{u\le t-1})\le p_{t-1}((S^n_u)_{u\le t-1}) +|\alpha_{t-1}|\sup_{u\le t-1}|S_u^n-S_u|+|\alpha_{t-1}|\epsilon.$$
As $\epsilon$ is arbitrarily chosen, we may conclude that 
$$p_{t-1}((S_u)_{u\le t-1})\le \liminf_n p_{t-1}((S^n_u)_{u\le t-1}).$$ 
\end{proof}

\subsection{Case where $x\mapsto g_t(S_0,\cdots,S_{t-1},x)$ is a convex function}

We shall prove that $p_{t-1}((S_u)_{u\le t-1})$ is a convex function of the price process $(S_u)_{u\le t-1}$ if so $ \Lambda_{t-1}$  is. In the following, we say that the mapping $$ \Lambda_{t-1}: (S_u)_{u \leq t-1} \longmapsto \Lambda_{t-1}((S_u)_{u \leq t-1}):=\overline{{\rm conv}}\left({\rm cl\,}(I_t((S_u)_{u\le t-1})|\cF_{t})\right)$$ is convex for the inclusion if, for $\lambda \in [0,1]$, 
$$\Lambda_{t-1}(( \lambda((S_u)_{u \leq t-1})+(1-\lambda)((\tilde{S}_u)_{u \leq t-1})  \subseteq \lambda \Lambda_{t-1}((S_u)_{u \leq t-1})+(1-\lambda) \Lambda_{t-1}((\tilde{S}_u)_{u \leq t-1}),$$ for all price process $(S_u)_{u \leq t-1},(\tilde{S}_u)_{u \leq t-1} $.
\begin{prop}\label{propconvex} Suppose that the mapping \[ (\omega,z) \mapsto g_t(S_0,S_1(\omega),...,S_{t-1}(\omega), z)
 \text{ is } \mathcal{F}_{t} \otimes \mathcal{B}(\R^d) \text{ measurable, }\]
non negative and
\[ z \mapsto g_t(S_0,S_1,...,S_{t-1}, z) \text{ is lower semi-continuous and convex almost surely } \]
and suppose that the mapping $\Lambda_{t-1}: (S_u)_{u \leq t-1} \longmapsto \Lambda_{t-1}((S_u)_{u \leq t-1})$  is convex. 
Then, the mapping $(S_u)_{u \leq t-1}\mapsto p_{t-1}((S_u)_{u \leq t-1})$ is convex .
\end{prop}
\begin{proof}
Let $\tilde{(S_u)}_{u\le t-1},(\overline{S_u})_{u\le t-1}$ be two price processes. Let us define the following price process $(S_u)_{u\le t-1}=\lambda (\overline{S_u})_{u\le t-1}+(1-\lambda) \tilde{(S_u)}_{u\le t-1}$ for $\lambda \in [0,1]$.
We consider the following random sets: 

\bean \Lambda_{t-1}&=&  \overline{{\rm conv}}\left({\rm cl\,}(I_t((S_u)_{u\le t-1})|\cF_{t})\right),\,t\ge 1,\\
\tilde{\Lambda}_{t-1}&= & \overline{{\rm conv}}\left({\rm cl\,}(I_t(\tilde{(S_u)}_{u\le t-1})|\cF_{t})\right),\,t\ge 1, \\
  \overline{\Lambda}_{t-1}&=&  \overline{{\rm conv}}\left({\rm cl\,}(I_t((\overline{S_u})_{u\le t-1})|\cF_{t})\right),\,t\ge 1.
\eean
 
  By assumption, we have $ \Lambda_{t-1} \subseteq \lambda \overline{\Lambda}_{t-1} +(1- \lambda) \tilde{\Lambda}_{t-1}$ for $\lambda \in [0,1]$.
 Let $\overline{h}$ and $\tilde{h}$ be two affine functions such that:
\[ \overline{h}(\overline{x}) \geq g_t((\overline{S_u})_{u \leq t-1},\overline{x}), \  \forall \overline{x} \in \overline{\Lambda}_{t-1}. \]
\[ \tilde{h}(\tilde{x}) \geq g_t((\tilde{S}_u)_{u \leq t-1},\tilde{x}), \  \forall \tilde{x} \in \tilde{\Lambda}_{t-1}. \]
Thus, for $\lambda \in ]0,1[$, we have
\bean
 \lambda \overline{h}(\overline{x})+ (1-\lambda) \tilde{h}(\tilde{x}) &\geq& \lambda g_t((\overline{S_u})_{u \leq t-1},\overline{x})+ (1-\lambda)g_t((\tilde{S}_u)_{u \leq t-1},\tilde{x})\\
 & \geq & g_t(\lambda((\overline{S_u})_{u \leq t-1} )+ (1- \lambda)((\tilde{S}_u)_{u\leq t-1} ) , \lambda \overline{x} + (1-\lambda) \tilde{x}). 
\eean
Let $x  \in \Lambda_{t-1}$ such that $x=\lambda \overline{x}+ (1-\lambda) \tilde{x}$.  By above, we have:
\bean
 \lambda \overline{h}(\overline{x})+ (1-\lambda) \tilde{h}(\tilde{x}) & \geq & g_t((S_u)_{u\leq t-1} ), x)=: \hat{g}_t(x).
\eean
Now, let us consider  $$E_x= \left\{ \frac{\lambda - 1 }{\lambda} \tilde{\Lambda}_{t-1}+ \frac{1}{\lambda}x, \lambda \in ]0,1[  \right\} \cap \overline{\Lambda}_{t-1}.$$

Observe that $\alpha E_{x_1}+(1-\alpha)E_{x_2}=E_{\alpha x_1+(1-\alpha)x_2}$ for all $\alpha\in [0,1]$, and $x_1,x_2\in \R^d$. Then, with $x=\alpha x_1+(1-\alpha)x_1$, any $\overline{x} \in E_x$ may be written as $\overline{x}=\alpha \overline{x}_1+(1-\alpha)\overline{x}_2$, where $\overline{x}_i \in E_{x_i}$, $i=1,2$. As $(x,\overline{x})\mapsto \tilde{h}(\frac{1}{1- \lambda}(x-\lambda \overline{x})) $ is affine, we deduce that 
\bean\lambda \overline{h}(\overline{x})+(1-\lambda)\tilde{h}(\frac{1}{1- \lambda}(x-\lambda \overline{x}))&\ge& \alpha \left(\lambda \overline{h}(\overline{x}_1)+(1-\lambda)\tilde{h}(\frac{1}{1- \lambda}(x_1-\lambda \overline{x}_1))\right)\\&&+(1-\alpha) \left(\lambda \overline{h}(\overline{x}_2)+(1-\lambda)\tilde{h}(\frac{1}{1- \lambda}(x_2-\lambda \overline{x}_2) )\right),\\
\lambda \overline{h}(\overline{x})+(1-\lambda)\tilde{h}(\frac{1}{1- \lambda}(x-\lambda \overline{x}))&\ge& \alpha \hat{h}(x_1)+(1-\alpha)\hat{h}(x_2),
\eean
where $\hat{h}(x)=\underset{\overline{x} \in E_x}\inf \lbrace \lambda \overline{h}(\overline{x})+(1-\lambda)\tilde{h}(\frac{1}{1- \lambda}(x-\lambda \overline{x})) \rbrace$. Therefore, taking the infimum in the right side of the inequality above, we deduce  that $\hat{h}$ is a (non negative) concave function with finite values. So, it is continuous and  we have $\hat{h}(x) \geq \hat{g}_t(x)$ for all $x \in \Lambda_{t-1}$. We deduce that
\bean
p_{t-1}((S_u)_{u \leq t-1}) & \leq & \hat{h}(S_{t-1})\\
& \leq & \lambda \overline{h}(\overline{S}_{t-1}) +(1-\lambda)\tilde{h}(\tilde{S}_{t-1}),\forall  \overline{S}_{t-1} \in \overline{\Lambda}_{t-1} ,\, \tilde{S}_{t-1} \in \tilde{\Lambda}_{t-1}.
\eean
Taking the infimum over all the affine functions $\overline{h}$ and $\tilde{h}$, we deduce that
\bean
p_{t-1}((S_u)_{u \leq t-1}) & \leq & \lambda p_{t-1}((\overline{S}_u)_{u \leq t-1}) +(1-\lambda)p_{t-1}((\tilde{S}_u)_{u \leq t-1}) 
\eean
and the conclusion follows.
\end{proof}

\begin{rem}\label{rem-continuous} Suppose that the AIP condition holds and that (\ref{hypothese}) holds. Consider $ \phi_{t-1}(u)= \underset{n }\inf\left\{\alpha^n u_{t-1}+f_{t-1}^*(-\alpha^n,u) \right\}$,  $u=(u_0,...,u_{t-1}) \in  (\R^d)^{t}$, where
$f_{t-1}^*(-\alpha,u) = \underset{m}\sup~ \left[ g_t(u,x^m(u))-\alpha x^m(u) \right]$. Recall that, by Proposition \ref{MeasPrice}, $p_{t-1}((S_u)_{u\le t-1})=\phi_{t-1}((S_u)_{u\le t-1}).$ When $g_t$ is convex, then  $\phi_{t-1}$ is convex by Proposition \ref{propconvex}. Moreover,  if $g_t\ge 0$,   $0 \leq \phi_{t-1}< \infty$ by  Proposition \ref{Bound+infprice=price}. Then, ${\rm dom}\, \phi_{t-1} = (\R^d)^t $ and we deduce that   $\phi_{t-1}$ is continuous on $(\R^d)^t $.
\end{rem}

\begin{rem}  Consider the case $d=1$. By a measurable selection argument, we may show that there exists $m_{t-1},M_{t-1}\in L^0([0,\infty],\cF_{t})$ such that 
$$\overline{{\rm conv}}\left({\rm cl\,}(I_t((S_u)_{u\le t-1})|\cF_{t})\right)=[m_{t-1},M_{t-1}].$$

By Lemma \ref{conc}, we deduce that  under (AIP)

\bea \label{Convex-price}p_{t-1}((S_u)_{u\le t-1})&=&g_t(S_0,\cdots,S_{t-1},m_{t-1})\\ \nonumber +&&\frac{g_t(S_0,\cdots,S_{t-1},M_{t-1})-g_t(S_0,\cdots,S_{t-1},m_{t-1})}{M_{t-1}-m_{t-1}}(S_{t-1}-m_{t-1}).
\eea

Moreover, the strategy is given by

$$\theta_{t-1}=\frac{g_t(S_0,\cdots,S_{t-1},M_{t-1})-g_t(S_0,\cdots,S_{t-1},m_{t-1})}{M_{t-1}-m_{t-1}}.$$

If we suppose that $m_{t-1}=k_{t-1}^dS_{t-1}$ and $M_{t-1}=k_{t-1}^uS_{t-1}$ as in \cite{BCL}, where $k_{t-1}^d$ and $k_{t-1}^u$ are deterministic coefficients, then $ p_{t-1}((S_u)_{u\le t-1})=g_{t-1}((S_u)_{u\le t-1})$ with 
$$g_{t-1}(x_0,\cdots,x_{t-1})=\lambda_{t-1}g_{t}(x_0,\cdots,x_{t-1},k_{t-1}^dx_{t-1})+(1-\lambda_{t-1})g_{t}(x_0,\cdots,x_{t-1},k_{t-1}^ux_{t-1}),$$
where $\lambda_{t-1}= \frac{k_{t-1}^u -1}{k_{t-1}^u-k_{t-1}^d}$ and $g_T$ is the payoff function.

At last, the order to be sent at time $t$ is given by the deterministic mapping defined on $\R^t$ by

$$\theta_{t-1}(s_0,\cdots,s_{t-1})=\frac{g_t(s_0,\cdots,s_{t-1},k_{t-1}^us_{t-1})-g_t(s_0,\cdots,s_{t-1},k_{t-1}^ds_{t-1})}{(k_{t-1}^u-k_{t-1}^d)s_{t-1}}.$$
\end{rem}

\begin{rem}[Market impact]

It is possible in our model to include a market impact. Indeed, it suffices to make the order (demand) mapping $ D_t(x)=\theta_t(x)-\theta_{t-1}(S_{t-1})$ coincided at time $t$ with the supply mapping $O_t(x)$, i.e. the available quantity  we may buy or sell at price $x$ in the order book. By convention, $O_t$ is negative for bid prices and positive for ask prices. It is an increasing function on $\R_+$ starting from $O_t(0+)=-\infty$ at price $0$ (we can sell as many assets as we want to the market at price $0$) and ending up with $O_t(+\infty)=+\infty$, i.e. we can buy as many assets as we want to the market at price $+\infty$. As soon as $D_t$ is bounded, there exists  executable bid prices $S^b_t$  in the order book such that $D_t(S^b_t)\ge O_t(S^b_t)$ when $D_t(S^b_t)\le 0$, i.e. the order may be executed  at price $S^b_t$ as the quantity $|D_t(S^b_t)|\le |O_t(S^b_t)|$. The executed bid price is naturally the best one among all possible. Similarly, there exists executable ask prices $S^a_t$ in the order book such that $D_t(S^a_t)\le O_t(S^a_t)$ when $D_t(S^a_t)\ge 0$ and  the order may be executed at price $S^a_t$ for the quantity $D_t(S^a_t)\le O_t(S^a_t)$. Note that the executed bid price may be closed to $0$ while the executed ask price may be very large. This liquidity phenomenon is then taken into account in the model through the conditional supports allowing to compute the strategy in our approach. 
\end{rem}

\subsection{The multistep backward procedure}

The main results of Section \ref{MR} for the one step model may be applied recursively, starting from time $T$, as the payoff function $g_T$ is known. 

Consider the case where the conditional support ${\rm cl\,}(I_t((S_u)_{u \leq t-1})|\cF_{t})$ admits a  Castaing representation $(\xi^m)_{m\ge 1}$ where $\xi^m=x^m((S_u)_{u \leq t-1})$, for all $m\ge 1$, and $x^m$ are Borel functions on $(\R^d)^t$. Then, by Proposition \ref{MeasPrice}, we know that the infimum price at time $T-1$ is a Borel function $g_{T-1}$ of the prices $S_0,\cdots,S_{T-1}$. Then, we may  repeat the procedure if we are in position to verify that $g_{T-1}$ is also l.s.c. This is the case by Proposition \ref{propconvex} and Remark \ref{rem-continuous}, under convexity conditions.   \smallskip

Many questions could be  investigated for future research, e.g. sensitivity to modeling assumptions, but also how to calibrate such a model from statistical estimations. Mainly, we need  to estimate conditional supports. This is illustrated in the numerical example that we propose in the next section. A technical question is also to consider discontinuous payoff functions even if this is less usual in finance where $g$ is generally a convex function. Actually, by Lemma \ref{conc}, we may replace the payoff function by its concave envelope. Note that our analysis is general enough to consider a lot of models, e.g. with  order books.

\section{Numerical illustration} \label{NI}

\subsection{Formulation of the problem with $d=1$}

In this section we consider the example of the European call option at time $T=2$, i.e. with the payoff function  $g(S_2)=(S_2-K)^+$, $K>0$. Let $(S_t)_{t=0,1,2}$ be the executed price process. Recall that $S_t$  belongs to the random set $\Lambda_t$, for $t=0,1,2,$ respectively. We suppose that the risk-free asset is given by $S^0=1$. Recall that there exist  $\mathcal{F}_{t+1}$-measurable closed random sets $I_t=I_t((S_u)_{u \leq t-1})$ such that: $$ \overline{\Sigma}(\Lambda_t((S_u)_{u \leq t-1}))= L^0(I_t((S_u)_{u \leq t-1}),\mathcal{F}_{t+1}),\quad t=0,1,2.$$
We may suppose that $\Lambda = \overline{\Sigma}(\Lambda)$ so that  $ S_t \in I_t$ a.s.  for $t=0,1,2$.
 At each step, we shall apply the procedure we have developed in the sections above. In particular, we seek for the strategy $\theta$ and we deduce the portfolio value $V$ associated to the executed price process $S$. Then, we may  estimate the error between the terminal value of $V_2$ and the payoff $g_2(S_2)$ that we denote by $ \epsilon_2= V_2-g_2(S_2)$. 
 
 We start from a known price $S_{-1}$ at time $t=0$, which corresponds to the last traded price. We suppose that  $I_t=I_t((S_u)_{u\le t-1})=[S_{t-1} m_t,S_{t-1} M_t]$, $t=0,1,2$, where the two random variables $m_t$ and $M_t$ are independent of $S_{t-1}$ and are uniformly distributed as $m_t \sim \mathcal{U}[0.7,1]$ and $M_t=m_t+spr_t$ such that $ spr_t\sim \mathcal{U}[0,0.4] $ is independent of $m_t$. Observe that $m_t^-=0.7$  and $M_t^+=1.4$.

At time $t=0$, we  choose in our model  to pick randomly $S_0$ in the interval $I_0$. Precisely, $ S_0 = S_{-1}m_0+k_0 S_{-1}(M_0-m_0)$, where $k_0$ is a random variable such that $k_0 \sim \mathcal{U}[0,1]$.  We make this choice for simplicity and that corresponds to the case where  the bid and ask prices of the market coincide with the mid price $S_0$. The order we sent is of the form {\sl buy or sell the quantity $\theta_0(z)$ at the price $z$.}

At time $t=1$, we choose to model   bid and  ask prices $S_1^{bid}$, $S_1^{ask}$ respectively as: $S_1^{bid}=S_0 m_1$ and $S_1^{ask}=S_0 M_1$ where $S_0$ is the last executed price. Notice that  the order of buying or selling depends on  the bid-ask values, see Figure  \ref{figure}. We define $S_1^*$ such that $\Delta \theta_1(S_1^*)=0$. If  $S^{bid}_1 \leq S^{ask}_1 \leq S_1^*$ (in the green zone $\{S_1:~\Delta \theta_1(S_1) \le 0\}$), then $S_1=S_1^{bid}$ since $\Delta \theta_1 \le 0$. If $S_1^*\le S^{bid}_1 \leq S^{ask}_1$, (the yellow zone,) then $S_1=S_1^{ask}$ as $ \Delta \theta_1> 0$. Otherwise, if $S^{bid}< S_1^* < S_1^{ask}$, we may arbitrarily choose  $S_1=S_1^{ask}$ or $S_1=S_1^{bid}$. In our model, we make the (arbitrary) choice that, if $|S_1^*-S_1^{bid}| \leq |S_1^*-S_1^{ask}|$, then $S_1= S_1^{ask}$ and $S_1= S_1^{bid}$ otherwise.

\begin{figure}[h!]
    \centering
        \includegraphics[scale=0.6]{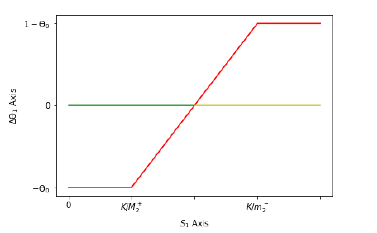} 
    \caption{}
    \label{figure}
\end{figure}

At last, we choose $S_2=S_2^{ask}=S_2^{bid} \in I_2=[m_2S_1, M_2S_1]$ accordingly to the formula $ S_2= S_1m_2+k_2 S_1 (M_2-m_2)$ where $k_2$ a uniform random variable in the interval $[0,1]$.

Note that the mapping $s_1\mapsto \Delta \theta_1(s_1)$ is the $\cF_1$-measurable order we send at time $t=1$, see Figure \ref{figure}. The later depends on $S_0$, which is $\cF_1$-measurable.

\subsection{Explicit computation of the strategy}

We deduce the portfolio value and the strategy value at any time by dominating the payoff function by the smallest affine function on the conditional support of $S$, as mentioned in (\ref{PriceInf}). We consider the terminal payoff function $g(S_T)=(S_T-K)^+$ for several strikes.

\subsubsection{The strategy at time $t=1$}
Recall that $S_2 \in \Lambda_2 (S_1) \sim I_2= [ S_1m_2, S_1M_2 ] $. In order to compute the strategy $\theta_1=\theta_1(S_1)$ we first compute the function $\varphi_1$ given by (\ref{PriceInf}) which dominates the the pay-off function $g_2$ on the conditional support  ${\rm cl\,}(I_2(S_1)|\cF_{2})=[S_1 m_2^-,S_1 M_2^+]$.

\subsubsection*{\textbf{1st case}: $K \in [S_1 m_2^-, S_1 M_2^+] \Leftrightarrow S_1 \in [\frac{K}{M_2^+}, \frac{K}{m_2^-}]$ .}

The dominating affine function $\varphi_1$, see Figure \ref{1s1}, is given by: $$ \varphi_1(x) = \frac{(S_1 M_2^+ - K)(x-S_1 m_2^-)}{S_1(M_2^+ - m_2^-)}.$$
So, $$V_1(S_1)=p_1(S_1)=\varphi_1(S_1)= \frac{(S_1 M_2^+ - K)(1- m_2^-)}{M_2^+ - m_2^-}=:g_1(S_1), $$ and 
 $$\theta_1(S_1)= \frac{S_1 M_2^+ - K}{S_1(M_2^+ - m_2^-)} .$$
 A simple computation shows that:
$$V_2= V_1(S_1)+ \theta_1(S_1)(S_2-S_1)=\varphi_1(S_2) \geq g_2(S_2).$$

%\begin{figure}[h!]
%\begin{center}
%\includegraphics[scale=0.5]{1s1.png}
%\caption{}
%\label{1s1}
%\end{center}
%\end{figure}

\begin{figure}[h!]
\begin{minipage}[c]{.38\linewidth}
\includegraphics[scale=0.5]{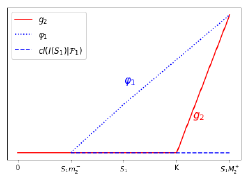}
\caption{}
\label{1s1}
 \end{minipage}% \hfill
 \begin{minipage}[c]{.38\linewidth} 
\includegraphics[scale=0.42]{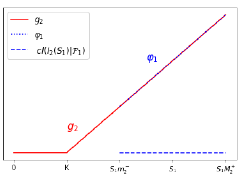} 
\caption{}
\label{2s1}
 \end{minipage} 
 \end{figure}

\subsubsection*{\textbf{2nd case}: $K \leq S_1  m_2^- \Leftrightarrow S_1 \geq  \frac{K}{m_2^-}$.}
In this case, we have $\varphi_1(x)=(x-K)^+$ for all $x \in [S_1 m_2^-, S_1 M_2^+]$, see Figure \ref{2s1}. 
Hence,  $V_1(S_1)= (S_1 -K)^+=:g_1(S_1)$ and $\theta_1(S_1)=1$. \\
%\begin{figure}[h!]
%\begin{center}
%\includegraphics[scale=0.5]{2s1.png} 
%\caption{}
%\label{2s1}
%\end{center}
%\end{figure}
\subsubsection*{\textbf{3rd case}: $K \geq S_1 M_2^+ \Leftrightarrow S_1 \leq \frac{K}{M_2^+}$.}
Observe that the dominating affine function $\varphi_1$ coincides with the x-axis on the support $[S_1 m_2^-, S_1 M_2^+]$, see Figure \ref{3s1}. Therefore, 
$V_1(S_1)=g_1(S_1):= 0 $ and we deduce that $\theta_1(S_1)=0$.
%\begin{center}
 %\includegraphics[scale=0.8]{3s1.png}
 %\end{center} 

 \begin{figure}[h!]

\includegraphics[scale=0.6]{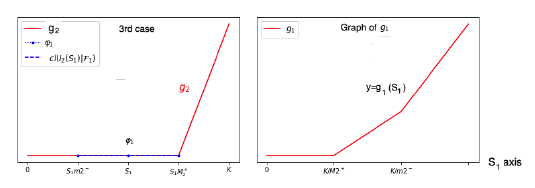}
\caption{}
\label{3s1}
 \end{figure}
\bigskip

We finally deduce that 
$$g_1(x)=\frac{(xM_2^+-K)(1-m_2^-)}{M_2^+-m_2^-}1_{\left[\frac{K}{M_2^+},\frac{K}{m_2^+}\right]}(x)+(x-K)^+1_{\left[\frac{K}{m_2^+},\infty\right)}(x).$$

The graph of the payoff function $g_1$  is represented in Figure \ref{3s1}.

\subsubsection{The strategy at time $t=0$}
In order to determine the strategy $\theta_0$, we compute the smallest affine  function  $\varphi_0$ that dominates $g_1$ on the conditional support $cl(I_1(S_0)| \mathcal{F}_0).$ 
\subsubsection*{\textbf{1st case}: $ S_0  M_1^+ \leq \frac{K}{M_2^+}$, i.e. $S_0\le \frac{K}{M_1^+M_2^+}$.}

We have $V_0(S_0)=g_0(S_0)= 0$ and $\theta_0(S_0)=0$, see Figure \ref{1s0}.

\begin{figure}[h!]
\begin{minipage}[c]{.38\linewidth}
\includegraphics[scale=0.5]{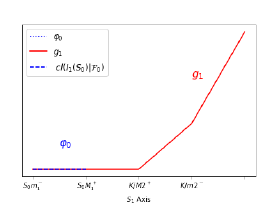}
\caption{}
\label{1s0}
 \end{minipage}% \hfill
 \begin{minipage}[c]{.38\linewidth} 
\includegraphics[scale=0.5]{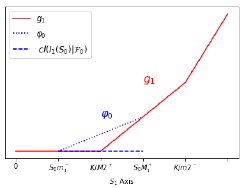} 
\caption{}
\label{2s00}
 \end{minipage} 
 \end{figure}

\subsubsection*{\textbf{2nd case}: $ S_0  m_1^- \leq \frac{K}{ M_2^+}$ and $ S_0  M_1^+ \in [\frac{K}{ M_2^+},\frac{K}{ m_2^-} ] $, i.e. $S_0\in [\frac{K}{ M_1^+M_2^+},\frac{K}{ m_1^- M_2^+}\wedge \frac{K}{ m_2^-M_1^+}]$.}

%\begin{center}
%\includegraphics[scale=0.8]{2s00.png} 
%\end{center}

We find that (see Figure \ref{2s00}):
\beqa
\varphi_0(x)= \frac{(S_0M_1^+M_2^+ -K)(1-m_2^-)}{S_0(M_1^+ - m_1^-)(M_2^+ - m_2^-)}(x-S_0m_1^-).
\eeqa
So, 
$$ V_0(S_0)=\varphi_0(S_0)= \frac{(S_0 M_1^+M_2^+ -K)(1 -  m_2^-)(1- m_1^-)}{( M_2^+ -  m_2^-)( M_1^+ -  m_1^-)}=:g_0(S_0),$$ and $$ \theta_0(S_0)= \frac{(S_0  M_1^+  M_2^+ -K)(1-  m_2^-)}{S_0( M_2^+ -  m_2^-)( M_1^+ -  m_1^-)}. $$

\subsubsection*{\textbf{3rd case}: $ S_0  m_1^- \leq \frac{K}{ M_2^+}$ and $ S_0  M_1^+ \geq \frac{K}{ m_2^-} $, i.e. $S_0\in [ \frac{K}{ m_2^-M_1^+},\frac{K}{ m_1^-M_2^+}]$.}

%\begin{center}
%\includegraphics[scale=0.8]{3s0.png} 
%\end{center}

We have, see Figure \ref{3s0}:
\beqa
\varphi_0(x)= \frac{S_0M_1^+ -K}{S_0(M_1^+ - m_1^-)}(x -S_0m_1^-).
\eeqa

\begin{figure}[h!]
\begin{minipage}[c]{.38\linewidth}
\includegraphics[scale=0.5]{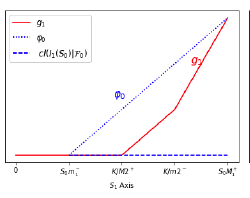}
\caption{}
\label{3s0}
 \end{minipage}% \hfill
 \begin{minipage}[c]{.38\linewidth} 
\includegraphics[scale=0.5]{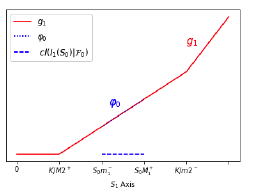} 
\caption{}
\label{4s0}
 \end{minipage} 
 \end{figure}

So,
\bean V_0(S_0)=\varphi_0(S_0)= \frac{(S_0 M_1^+ -K)(1- m_1^-)}{ M_1^+ - m_1^-}=:g_0(S_0),\quad \theta_0(S_0)= \frac{S_0  M_1^+ - K}{S_0(  M_1^+ -  m_1^-)}.\eean

\subsubsection*{\textbf{4th case}: $ S_0  m_1^- \in [ \frac{K}{ M_2^+},\frac{K}{ m_2^-} ] $ and $ S_0  M_1^+ \in [\frac{K}{ M_2^+},\frac{K}{ m_2^-} ] $, i.e. $S_0\in [\frac{K}{ m_1^-M_2^+},\frac{K}{ m_2^- M_1^+}]$.}

%\begin{center}
%\includegraphics[scale=0.8]{4s0.png} 
%\end{center}

We have $\varphi_0(x)=g_1(x)$, for all $x \in cl(I_1(S_0)| \mathcal{F}_0)$, see Figure \ref{4s0}. Therefore, 

$$V_0(S_0)=\varphi_0(S_0)= \frac{(S_0 M_2^+ -K)(1 -  m_2^-)}{ M_2^+ - m_2^-}=:g_0(S_0),\quad \theta_0(S_0)= \frac{  M_2^+(1-m_2^-)}{ M_2^+  -  m_2^-}.$$

\subsubsection*{\textbf{5th case}: $ S_0  m_1^- \in [ \frac{K}{ M_2^+ },\frac{K}{ m_2^-} ] $ and $ S_0  M_1^+ \geq { \frac{K}{ m_2^-} }$, i.e. $S_0\in [\frac{K}{m_1^- M_2^+ }\vee \frac{K}{ m_2^-M_1^+},\frac{K}{ m_1^-m_2^-} ]$.}

%\begin{center}
%\includegraphics[scale=0.8]{5s0.png} 
%\end{center}

We obtain that (see Figure \ref{5s0}):

\bean \varphi_0(x)&=&\frac{(S_0M_1^+ -K)(M_2^+-m_2^-)-(S_0m_1^-M_2^+ -K)(1-m_2^-)}{S_0(  M_1^+ -  m_1^-)(M_2^+ -  m_2^-)}x \\
&& + \frac{-m_1^-(S_0M_1^+ -K)(M_2^+-m_2^-)+M_1^+(S_0m_1^-M_2^+ -K)(1-m_2^-)}{(  M_1^+ -  m_1^-)(M_2^+ -  m_2^-)}.\eean

\begin{figure}[h!]
\begin{minipage}[c]{.38\linewidth}
\includegraphics[scale=0.5]{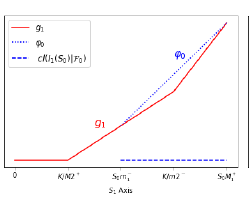}
\caption{}
\label{5s0}
 \end{minipage}% \hfill
 \begin{minipage}[c]{.38\linewidth} 
\includegraphics[scale=0.5]{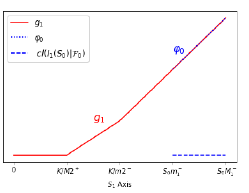} 
\caption{}
\label{6s0}
 \end{minipage} 
 \end{figure}

Then,
\bean V_0(S_0)&=&\varphi_0(S_0)=:g_0(S_0)\\
&&= \frac{(S_0 M_1^+ -K)(M_2^+  -  m_2^-)(1- m_1^-)-(S_0  m_1^- M_2^+ -K)(1- m_2^-)( 1- M_1^+ )}{(M_1^+ - m_1^-)(M_2^+ -  m_2^-)}\eean and $$~~~~~~~~ \theta_0(S_0)=\frac{(S_0 M_1^+ - K)(M_2^+ - m_2^-)- (S_0 m_1^- M_2^+ -K)(1-m_2^-)}{S_0 (M_2^+ - m_2^-) (M_1^+ - m_1^-)}. $$

\subsubsection*{\textbf{6th case}: $ S_0  m_1^- \geq { \frac{K}{ m_2^-} } $ and $ S_0  M_1^+ \geq { \frac{K}{ m_2^-} } $, i.e. $S_0\ge \frac{K}{ m_2^-m_1^-}$.}

%\begin{center}
%\includegraphics[scale=0.8]{6s0.png} 
%\end{center}

We have $V_0(S_0)= (S_0-K)^+=:g_0(S_0)$ and $\theta_0(S_0)=1$, see Figure \ref{6s0}. \bigskip

\subsection{\textbf{Empirical results}}

For an observed price $S_{-1}$ at time $t=0$ (which corresponds to the last traded price), and for different strike values $K$, we test the infimum super-hedging strategy by computing the relative error $\epsilon_R$ from a data set of $10^6$ simulated prices $S_t$ for $ t \in {0,1,2}.$ To do so, we wrote a script in Python. The relative error is given by $$\varepsilon_R= \frac{V_2-(S_2-K)^+}{S_2} .$$

In the following table \ref{TabResults},  empirical  results are presented for different  values of the strike $K$ and a sample of $10^6$ scenarios.

\begin{figure}[h!] \label{TabResults}
\begin{center}

\begin{tabular}{|l|c|c|c|c|l|} %\label{TabResults}
   \hline
   $K $ & $ 50$ & $75$ & $100$ & $125$ & $150$\\
  \hline
  \hline
   
    $E(S_0)$& 95.002& 94.983 &  95.006& 94.98 & 95.001 \\
  \hline
    $E(S_1)$& 99.56& 94.94 &  87.085& 82.104 & 81.736 \\
 \hline
    $E(S_2)$& 94.56& 90.180 &  82.716& 78.01 & 77.664 \\   
  \hline
  \hline
  $E(V_0) $&46.503 &29.357 & 16.960& 11.244& 6.7 \\
  \hline
    $\max V_0$ & 89.677 &  66.72 & 49.726&  33.05 &  22.562 \\
  \hline
  \hline
        $E(V(S_0)/S_{-1})$ & 0.465&0.294&  0.170 & 0.112& 0.067 \\
  \hline
    $E(V(S_0)/S_{0})$ &0.483 &  0.300 &  0.173 & 0.114& 0.066 \\
  \hline
   $\min(V(S_0)/S_0)$& 0.359 & 0.163  &  0.098& 0.032& 0 \\
  \hline
   $\max(V(S_0)/S_0)$ & 0.642& 0.479 & 0.358&  0.237 & 0.162\\
  \hline
  \hline
  $E(\epsilon_R)$ &0.017 & 0.077 & 0.076 & 0.064&   0.039 \\
  \hline
   $\sigma(\epsilon_R)$ &0.024  & 0.045 & 0.04& 0.037 &0.0317 \\
  \hline  
  $ \min\epsilon_R$ & $0$ & $2.23 * 10^{-6} $&  $1,9 * 10^{-7}$ &$ 5.975* 10^{-8} $& $0$ \\
  \hline   
  $\max(\epsilon_R)$ &0.18  & 0.19 & 0.195& 0.187 &0.187 \\
\hline
\hline
$E(\theta_0S_0/V_0)$ &199\%  & 255\% & 322\%& 333\% &313\% \\
\hline
$E(\theta_1S_1/V_1)$ &205\%  & 230\% & 134\%& 32\% &3\% \\
\hline
\end{tabular}
\end{center}
\caption{The empirical results.}
\label{TabResults}
\end{figure}

 We observe that the executed prices depend on the strike $K>0$, i.e. there is a market impact of the orders on the prices. Indeed, as expected, the orders we send depend on the payoff function. As $K$ increases, the payoff decreases and, as expected, the option price $V_0$ decreases. The distribution of $S_1$ admits two regimes as seen in Figure \ref{DistriS1} that correspond to the bid and ask prices.

Notice that the  proportion of the portfolio value invested in the risky assets at time $t=1$ decreases as the payoff decreases. We also observe that this proportion decreases (resp. increases) when the price $S$ decreases (resp. increases) between time $t=0$ and $t=1$, i.e. when $\Delta S_1<0$ (resp. $\Delta S_1\ge 0$). At last, the empirical results obtained for the relative error confirm  the efficiency of the super-hedging strategy, see Figure \ref{DistriErr}.

\begin{figure}[h!]
\begin{minipage}[c]{.38\linewidth}
\includegraphics[scale=0.5]{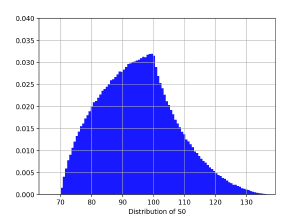}
\caption{}
\label{DistriS0}
 \end{minipage}% \hfill
 \begin{minipage}[c]{.38\linewidth} 
\includegraphics[scale=0.5]{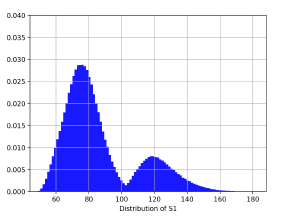} 
\caption{K=100.}
\label{DistriS1}
 \end{minipage} 
 \end{figure}

%\begin{center}
%\includegraphics[scale=0.5]{distribS0.png} 
%\end{center}

%\begin{center}
%\includegraphics[scale=0.5]{distribS1.png} 
%\end{center}

\begin{figure}[h!]
\begin{minipage}[c]{.38\linewidth}
\includegraphics[scale=0.5]{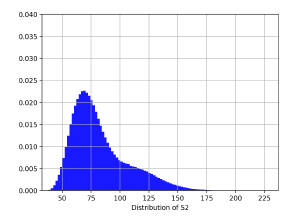}
\caption{K=100.}
\label{DistriS2}
 \end{minipage}% \hfill
 \begin{minipage}[c]{.38\linewidth} 
\includegraphics[scale=0.5]{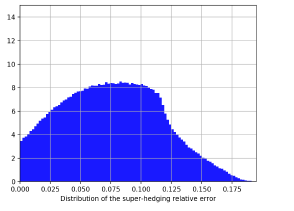} 
\caption{K=100.}
\label{DistriErr}
 \end{minipage} 
 \end{figure}

%\begin{center}
%\includegraphics[scale=0.5]{distribS2.png} 
%\end{center}

%\begin{center}
%\includegraphics[scale=0.5]{distribErr.png} 

%\end{center}


\bigskip




\newpage

\begin{thebibliography}{100}

\bibitem{AO} Agram N. and \O ksendal B. A financial market with singular drift and no arbitrage. Mathematics and Financial Economics, \url{ https://doi.org/10.1007/s11579-020-00284-9}, 2020.

\bibitem{BCL} Baptiste J., Carassus L. and L\'epinette E. Pricing without martingale measures. Preprint. \url{https://hal.archives-ouvertes.fr/hal-01774150}.

 \bibitem[Barron, E. N. and Cardaliaguet, P. and Jensen, R.]{BC}
Baron, Cardaliaguet P. and Jensen, R. Conditional essential suprema with applications. Applied Mathematics and Optimization, 48 (2003),  3, 229-253. 

\bibitem{BK} Becherer D. and Kentia K. Good deal hedging and valuation under combined uncertainty about drift and volatility. Probability, Uncertainty and Quantitative Risk, 2, 1, 1-40, 2017.

\bibitem{BK1} Becherer D. and Kentia K. Hedging under generalized good-deal bounds and model uncertainty. Mathematical Methods of Operations Research, 86, 1, 171-241, 2017.

\bibitem{BGK}Bertsimas D., Gupta V. and Kallus N. Data-driven robust optimization. Mathematical Programming, 12, 167, 2017.



\bibitem{BBK} Biagini S., Bouchard B. and Kardaras C. Robust fundamental theorem for continuous processes. Mathematical Finance, 27, 4, 963-987, 2017.

\bibitem{BN} Bouchard B. and Nutz M. Arbitrage in nondominated discrete-time models. The Annals of Applied Probability, 25, 2, 823-859, 2015.

\bibitem{BFM} Burzoni M, Fritelli M.,  Maggis M. Universal arbitrage aggregator in discrete-time markets under uncertainty. Finance and Stochastics, 20, 1, 1-50, 2016.

\bibitem{BRS} Burzoni M, Riedel F. and Soner M.H. Viability and arbitrage under Knightian uncertainty. Swiss Finance Institute Research Paper, 17-48, 2018.

 \bibitem{CL}Carassus L. and L\'epinette E. Pricing without no-arbitrage condition in discrete-time.  Journal of Mathematical Analysis and Applications, 505, 1, 125441, 2021. 


\bibitem{COW}Carassus L., Obl\`oj J. and Wiesel J. The robust superreplication problem: a dynamic approach. SIAM J. Financial Math., 10, 4, 907-941, 2019.

\bibitem{CKT}Cheredito P, Kupper M. and Tangpi L. Duality formulas for robust pricing and hedging in discrete time. SIAM J. Financial Math., 8, 1, 738-765, 2017.


\bibitem{CuKT} Cuchiero C, Klein I. and Teichmann J. A fundamental theorem of asset pricing for continuous time large financial markets in a two filtration setting. Theory of Probability and its Applications, 65, 3, 2020.

\bibitem{Da} Dahl K.R. Pricing of claims in discrete time with partial information. Applied Mathematics and Optimization, 68, 145-155, 2013.

\bibitem{DKS} De Vali\`ere D., Stricker C. and Kabanov Y. No-arbitrage properties for financial markets with transaction costs and incomplete information. Finance and Stochastics, 11, 2, 237-251, 2007.



\bibitem{EL} El Mansour M. and L\'epinette E. Conditional interior and conditional closure of a random set. Journal of Optimization and Applications, 187, 356-369, 2020.

\bibitem{FNS}Fadina T, Neufeld A. and Schmidt T.  Affine processes under parameter uncertainty. Probability, Uncertainty and Quantitative Risk, 4, 5, 2019.



\bibitem{Hess} Hess C., Seri R., and Choirat C. Set-valued integration and set-valued probability theory: an overview. E. Pap, editor, Handbook of Measure Theory, 14, 617-673. Elsevier, 2002.

\bibitem{H} Hobson D.G. Robust hedging of the lookback option. Finance and Stochastics, 2, 4, 329-347, 1998.

 \bibitem{KS} Kabanov Y., Safarian M. Markets with transaction costs. Mathematical Theory. Springer-Verlag Berlin Heidelberg, 2009.
 
 \bibitem{KSW} Kazmerchuk Y., Swishchuk A. and Wu Jianhong. The pricing of options for securities markets with delayed response. Mathematics and computers in simulation, 75, 69-79, 2007.
 
 \bibitem{IM} Ichiba T. and Mousavi M. Option pricing with delayed information. Article online: \url{https://www.researchgate.net/publication/318256675_Option_Pricing_with_Delayed_Information}, 2017.
 
 \bibitem{KL}Kabanov Y., L\'epinette E. Essential supremum with respect to a random partial order. Journal of Mathematical Economics, 49 (2013), 6, 478-487.   
 
 
\bibitem{Kn} Knight F. H. Risk, Uncertainty, and Profit. Boston MA: Hart, Schaffner and Marx; Houghton Mifflin, 1921.

\bibitem{La} Lackner P. Utility maximization with partial information. Stochastic Processes and their Applications, 56, 2, 247-273.

  \bibitem{LM} L\'epinette E., Molchanov I. Conditional cores and conditional convex hulls of random sets. Journal of Mathematical Analysis and Applications,  478 (2019), 2, 368-392.
\ bibitem{LM1} L\'epinette E., Molchanov I. Risk arbitrage and hedging to  acceptability. Finance and Stochastics, 25 (2021),101-132.
 
 \bibitem{NN} Neufeld A. and Nutz M. Superreplication under volatility uncertainty for measurable claims. Electronic Journal of Probability, 18, 48, 2013.
 
 \bibitem{OSZ} \O ksendal B., Sulem A. and Zhang T. Optimal control of stochastic delay equations and time-advanced backward stochastic differential equations. Advances in Applied Probability, 43, 2, 572-596, 2011.

 \bibitem{OW} Obl\'oj Jan and Wiesel J. Robust estimation of superhedging prices. Annals of Statistics, 49, 1, 508-530, 2021.
 
 \bibitem{PWZ}Pham H., Wei X. and Zhou C. Portfolio diversification and model uncertainty: a robust dynamic mean-variance approach. Article on ArXiv: \url{https://arxiv.org/abs/1809.01464}, 2018.

\bibitem{PR} Platen E. and Rungaldier W.J. A benchmark approach to portfolio optimization under partial information. Asia-Pacific Financial Markets, 14, 25-43, 2007.

\bibitem{Q} Quenez M-C. Optimal portfolio in a multiple-priors model. Seminar on Stochastic Analysis, Random Fields and Applications IV, Progr. Probab., 58, 291-321, Birkhuser, Basel, 2004.

\bibitem{RM} R\'asonyi M. and Meireles-Rodrigues A. On utility maximisation under model uncertainty in discrete-time markets. Mathematical Finance, 31, 1, 149-175, 2021.

 \bibitem[Rockafellar and Wets(1998)]{rw}
R.T. Rockafellar and R.J.-B. Wets. Variational analysis, 317, Grundlehren der
  Mathematischen Wissenschaften: Fundamental Principles of Mathematical
  Sciences. Springer-Verlag, Berlin, 1998.
  
  \bibitem{SZ} Saporito Y.F. and Zhang J. Stochastic control with delayed information and related nonlinear master equation. Siam Journal on Control and Optimization, 57, 1, 693-717, 2019.

\bibitem{TTU} Tevzaze R., Toronjadze T. and Uzunashvili T. Robust utility maximization for a diffusion market model with misspecified coefficients. Finance and Stochastics, 17, 3, 535-563, 2009.

\end{thebibliography}
\end{document}